%% file: main.tex
\documentclass[preprint,times]{elsarticle}
\journal{Systems \& Control Letters}

\usepackage{lineno}

\usepackage[OT1]{fontenc}
\usepackage{preamble}
\input{shortcuts.tex}

\newcommand{\yang}[1]{\textcolor{ForestGreen}{[\textbf{Yang:} #1]}}
\newcommand{\lina}[1]{\textcolor{red}{[\textbf{Lina:} #1]}}
\newcommand{\rev}[1]{#1}  

\begin{document}

  \allowdisplaybreaks[4]
  \setlength{\abovedisplayskip}{1pt}
  \setlength{\abovedisplayshortskip}{1pt}
  \setlength{\belowdisplayskip}{1pt}
  \setlength{\belowdisplayshortskip}{1pt}
  \setlength{\jot}{1pt}  
  \setlength{\floatsep}{1ex}
  \setlength{\textfloatsep}{1ex}


    
  \input{00-frontmatters}
  \input{01-introduction}
  \input{02-settings}
  \input{03-results}
  \input{04-simulations}
  \input{05-backmatters}

  \bibliographystyle{elsarticle-num-names} 
  \bibliography{ref}


  \appendix
  \input{A1-appendix}
  \vfill

\end{document}

%% file: shortcuts.tex
\DeclareMathAlphabet{\bdxscr}{U}{BOONDOX-cal}{m}{n}
\SetMathAlphabet{\bdxscr}{bold}{U}{BOONDOX-cal}{b}{n}
\newcommand{\g}{\bdxscr{g}}

\newcommand{\boldtitle}[1]{\medskip\noindent\textbf{#1}}

\renewcommand{\preceq}{\preccurlyeq}
\renewcommand{\succeq}{\succcurlyeq}

\makeatletter
\newcommand{\oset}[3][0ex]{%
  \mathrel{\mathop{#3}\limits^{%
  \vbox to#1{\kern-2\ex@%
  \hbox{$\scriptstyle#2$}\vss}}%
}}
\makeatother
\newcommand{\todist}{\mathrel{\oset[0.45ex]{\mathrm{d}}{\to}}}

\newcommand{\sddots}{\raisebox{0pt}{\scalebox{.66}{$\ddots$}}}

\newcommand{\zero}{\bm{0}}

\newcommand{\Sg}{\varSigma}
\newcommand{\As}{A_{\star}}
\newcommand{\Ss}{\varSigma_{\star}}
\newcommand{\thetas}{\theta_{\star}}

\newcommand{\hA}{\hat{A}}
\newcommand{\hS}{\hat{\varSigma}}
\newcommand{\htheta}{\hat{\theta}}
\newcommand{\hAols}{\hat{A}_{\mathrm{ols}}}
\newcommand{\hSols}{\hat{\varSigma}_{\mathrm{ols}}}
\newcommand{\hAmle}{\hat{A}_{\mathrm{mle}}}
\newcommand{\hSmle}{\hat{\varSigma}_{\mathrm{mle}}}
\newcommand{\hAsme}{\hat{A}_{\mathrm{sme}}}
\newcommand{\hSsme}{\hat{\varSigma}_{\mathrm{sme}}}

\renewcommand{\P}{\mathbb{P}}

\newcommand{\D}{\mathcal{D}}
\newcommand{\dD}{d^{\mathcal{D}}}
\renewcommand{\F}{\mathcal{F}}
\renewcommand{\H}{\mathcal{H}}
\renewcommand{\S}{\mathcal{S}}

\renewcommand{\L}{\mathcal{L}}
\newcommand{\hL}{\hat{\mathcal{L}}}

\newcommand{\Lmle}{\mathcal{L}_{\mathrm{mle}}}
\newcommand{\hLmle}{\hat{\mathcal{L}}_{\mathrm{mle}}}
\newcommand{\Lsme}{\mathcal{L}_{\mathrm{sme}}}
\newcommand{\hLsme}{\hat{\mathcal{L}}_{\mathrm{sme}}}

\newcommand{\lolambda}{\lambda_{-}}
\newcommand{\hilambda}{\lambda_{+}}

\newcommand{\tw}{\tilde{w}}

\usepackage{environ}

\NewEnviron{equationfit}[1][0.9\linewidth]{%
\begin{equation*}
  \resizebox{#1}{!}{\(
    \BODY
  \)}
\end{equation*}
}
\NewEnviron{alignfit}[1][0.9\linewidth]{%
\begin{equation*}
  \resizebox{#1}{!}{\(
  \begin{aligned}
    \BODY
  \end{aligned}
  \)}
\end{equation*}
}

%% file: 00-frontmatters.tex
\begin{frontmatter}

  \title{Joint Identification of Linear Dynamics and Noise Covariance via Distributional Estimation}
  \author[Harvard]{Yang Hu\corref{corr}}
  \author[Harvard]{Na Li}
  \affiliation[Harvard]{%
    organization={Harvard University},%
    addressline={150 Western Ave},%
    city={Boston},%
    postcode={02134},%
    state={MA},%
    country={USA}%
  }
  \cortext[corr]{%
    Yang Hu (email: \url{yanghu@g.harvard.edu}) is the corresponding author of this paper.%
  }

  
  \begin{abstract}
    In this paper, we propose a novel framework for the joint identification of system dynamics and noise covariance in linear systems, under general noise distributions beyond Gaussian. Specifically, we would like to simultaneously estimate the dynamical matrix $A$ and the noise covariance matrix $\Sg$ using state transition data. The formulation builds upon a novel parameterization of the state-transition distribution, which enables more effective use of distributional “shape” information for improved identification accuracy. We introduce two practical estimators, namely the maximum likelihood estimator (MLE) and the score-matching estimator (SME), to solve the joint dynamics-covariance identification problem, and provide rigorous analysis of their statistical properties and sample complexity. Simulation results show that the proposed estimators outperform the ordinary least squares (OLS) baseline.
  \end{abstract}
  
  
  
  \begin{keyword}
    linear systems \sep system identification \sep noise covariance identification \sep maximum likelihood estimator \sep score matching \sep sample complexity
  \end{keyword}

\end{frontmatter}

%% file: 01-introduction.tex
\section{Introduction}\label{sec:1-introduction}

System identification---the problem of estimating system dynamics using state transition data---has been a long-standing research field with broad applications, ranging from control and robotics to econometrics and neuroscience. In the past, classical theoretical results in this field mainly focused on the consistency (i.e., asymptotic convergence) of estimators, including ordinary least-squares (OLS) estimators \cite{ljung1976consistency} and maximum likelihood estimators (MLE) \cite{ljung1976on}. The more recent advances in learning theory have brought significant progresses in establishing finite-time convergence guarantees for these estimators \cite{simchowitz2018learning, sarkar2019near, oymak2019non, jedra2022finite}. These developments have, in turn, laid foundations for finite-time regret analyses in downstream control synthesis tasks \cite{cohen2019learning, mania2019certainty, dean2020sample, simchowitz2020naive}.

However, these finite-time analyses of linear system identification have primarily focused on estimating the nominal system dynamics (e.g., the dynamic matrix $
A$ for autonomous linear systems), typically under strong simplifying assumptions on the noise model like isotropic Gaussian or sub-Gaussian noise. Such assumptions do facilitate tractable analysis and favorable sample complexity guarantees, though at the risk of over-simplifying real-world systems by ignoring the crucial “shape” information hidden in the noise distribution, such as its tail behavior and higher-order moments. For instance, the success of OLS-based estimators in identifying nominal dynamics is known to implicitly rely on the light-tailed nature of the noise, which may not hold in practice and could lead to performance degradation. 

Beyond the question of whether such noise assumptions are realistic or not, another important yet often overlooked issue is the identification of the noise covariance structure itself. Even under the classical assumption of light-tailed Gaussian noise, estimating the noise covariance remains a meaningful and practically relevant task. Such information plays a pivotal role in a variety of downstream applications, including distributionally robust control \cite{zymler2013distributionally, van2015distributionally}, risk-sensitive control \cite{kishida2023risk, hu2024risk}, and adaptive Kalman filtering \cite{huang2017novel, limaverde2022adaptive}, where knowledge of the noise covariance is typically assumed a priori. Therefore, there is a clear need to move beyond merely identifying the nominal dynamics and to develop efficient methods for jointly identifying both the system dynamics and the underlying noise structure.



These observations have recently spurred growing interest in noise modeling and identification.  Following the pioneering work \cite{mehra1970identification}, multiple covariance identification methods have been designed in the past decade, including auto-covariance least-squares (ALS) methods \cite{odelson2006new, ge2017noise} and MLE methods \cite{bavdekar2011identification, zagrobelny2015identifying}. However, these methods only focus on the noise covariance identification and generally assume known system dynamics. In brief, while established methods exist for estimating system dynamics or noise covariance separately, principled approaches for the joint identification of both are still lacking.  

One intuitive approach to  bridge the above gap between system identification and noise covariance identification is to resort to a two-stage scheme---first identifying the nominal dynamics, and then recovering the noise covariance using the identified dynamics. However, as discussed later in this paper, such a tandem approach is technically equivalent to assuming Gaussian noise, thereby failing to capture the richer “shape” information presented in more general noise models. These challenges motivate the following research question:

\vspace{-2mm}
\begin{center}
  \textit{Is it possible to perform joint identification of dynamics\\and noise covariance for generic noise distribution?}
\end{center}
\vspace{-2mm}

\noindent Here by ``joint identification'' we explicitly mean to rule out the tandem approach stated above. Instead, we would like to reparameterize the problem to encapsulate both the dynamics and the covariance into a uniform set of parameters, hence enabling algorithms that effectively and efficiently estimate the parameters for the joint identification. As we will see later, an appropriate generic parameterization scheme is critical to this task.

\boldtitle{Contributions.} In this paper, we propose to model noise distribution as a unit-covariance base distribution scaled by the covariance matrix, such that the state transition distribution belongs to a density family parameterized by the dynamical matrix $A$ and the scale matrix $\Sg$. Such reparameterization is favored in that it is general enough to include a wide range of known noise distributions, and that it directly encodes $A$ and $\Sg$ in the parameters to reformulate the joint identification task as a parameter estimation problem. Based on such reparameterization, we propose two novel estimators, i.e. the maximum likelihood estimator (MLE) and the score-matching estimator (SME), that outperform the naive ordinary least-squares (OLS) estimator via theoretical arguments and empirical evidence.

The contribution of this paper is two-fold.
\begin{enumerate}
  \item We establish the connection between state transition in linear systems subject to generic stochastic noise and $\phi$-density families, and highlight an important sub-collection of elliptical families that greatly facilitates computation. As far as we know, this is the first paper to study the problem of joint identification of linear dynamics and noise covariance (the ``\emph{joint identification}'' problem) and propose an appropriate model to reformulate it as a parameter estimation problem.
  
  \item We propose two novel estimators, namely MLE and SME, to solve the joint identification problem, along with a characterization of their key properties and a preliminary sample complexity analysis. The performance of the estimators is compared against the naive OLS baseline and justified by empirical evaluation results.
\end{enumerate}

The rest of this paper is organized as follows. In \Cref{sec:2-settings} we introduce the system setting and problem formulation, as well as preliminaries on $\phi$-density families. In \Cref{sec:3-algorithm} we define the proposed estimators and characterize their theoretical properties, and also provide a sample complexity analysis. In \Cref{sec:4-simulations} we include simulation results that further justify the performance of the proposed estimators.

%% file: 02-settings.tex
\section{Problem Formulation and Preliminaries}\label{sec:2-settings}

In this section, we first formulate our system setting and the joint identification problem, and then provide some necessary preliminaries on the candidate distribution families.

\subsection{Problem Formulation}

\boldtitle{System Setting.} Consider a linear time-invariant (LTI) system
\begin{equation}\label{eq:system_dynamics}
  x' = \As x + \Ss^{1/2} w,
\end{equation}
where $x,x' \in \R^d$ denote the current and next \emph{state} of the system, respectively; $\As \in \R^{d \times d}$ is the \emph{dynamical} matrix, and $\Ss \in \S_{\succ \zero}^{d}$ is the \emph{scale} matrix that regulates the covariance of the effective noise input $\Ss^{1/2} w$, both of which are unknown and to be estimated; $w \in \R^d$ is an i.i.d. \emph{unit noise} with zero mean ($\E{w} = \zero$) and unit covariance ($\Var{w} = I$). We emphasize that, as discussed in the introduction, we would like to handle \emph{general} noise distributions, rather than being confined to any specific choices like Gaussian noise. Specifically, we assume a known \emph{base density} $\phi(\cdot)$ of the unit noise $w$, which is allowed to take \emph{any} form, as long as it satisfies the following regularity assumptions.

\begin{assumption}[Base density]\label{assum:diff_base_density}
  The base density $\phi$ is positive and twice continuously differentiable, and induces a finite first-order moment. Further, denote the gradient and Hessian of $\log \phi(w)$ by $\g(w)$ and $\H(w)$, respectively; i.e.,
  \begin{align*}
    \g(w) &:= \nabla_{w} \log \phi(w) = \frac{\nabla_{w} \phi(w)}{\phi(w)},\\
    \H(w) &:= \nabla^2_w \log \phi(w) = \nabla_w \g(w).
  \end{align*}
\end{assumption}

As will be seen later, the parameter estimators generally involve log-densities and their derivatives, and the assumption is hence made to facilitate the design of such estimators.

\boldtitle{The Joint Identification Problem.} On a high level, the goal of the proposed \emph{joint identification} problem is to identify the dynamical and scale matrices $(A, \Sg)$, using state transition data $\D = \set{(x_i, x'_i) \mid i \in [N]}$. Note that, in general, we don't make assumptions on the dataset $\D$, which is allowed to contain both single- and multiple-trajectory data.

A natural and concrete way to solve this joint identification problem is to convert it into a parameter estimation problem, which leads to various well-known estimators including the maximum likelihood estimator (MLE) and the score-matching estimator (SME) proposed below. Specifically, for any loss function $\ell(x, x'; A, \Sg)$ penalizing ``less likely'' transitions from $x$ to $x'$ under parameters $(A, \Sg)$, the corresponding estimator is
\begin{equation}
  (\hA, \hS) := \arg\min_{A,\Sg}\; \sum_{i=1}^{N} \ell(x_i, x'_i; A, \Sg).
\end{equation}
The high-level goal of the remaining paper is to construct proper losses $\ell$ and characterize the corresponding estimators in terms of statistical properties including sample complexity.

\subsection{Preliminaries}

In this section, we introduce some preliminaries on how to properly reparameterize the noise distribution. We first show that the transition distribution implied by \eqref{eq:system_dynamics} belongs to a family determined by the base density $\phi$, and then introduce \emph{elliptical} family as a computation-friendly special case. For the sake of presentation, we only focus on distributions with densities.

\boldtitle{Reparameterizing as $\bm{\phi}$-Density Family.} We point out again that various methods have been proposed in literature to perform either \emph{dynamics identification} (estimate $A$ given known $\Ss$) \cite{simchowitz2018learning, sarkar2019near, oymak2019non, jedra2022finite} or \emph{covariance identification} (estimate $\Sg$ given known $\As$) \cite{odelson2006new, ge2017noise, bavdekar2011identification, zagrobelny2015identifying} separately using state transition data, but not for \emph{joint identification} that simultaneously identifies $(A, \Sg)$ (both assumed to be unknown). A major issue that obstructs the application of data-driven methods for joint identification is that the roles played by the $A$ and $\Sg$ matrices in state transitions are quite distinct. Indeed, conditioned on any fixed current state $x$, the dynamics \eqref{eq:system_dynamics} guarantees that the conditional expectation and variance of the next state $x'$ are, respectively,
\begin{align*}
  \Es{x' \mid x} &= \Es{\As x \mid x} + \Es{\Ss^{1/2} w \mid x} = \As x,\\
  \Vars{x' \mid x} &= \Vars{\As x \mid x} + \Vars{\Ss^{1/2} w \mid x} = \Ss.
\end{align*}
Briefly put, the expectation only depends on $\As$, while the variance is fully determined by $\Ss$, which indicates that they have to be handled independently without a unified reparameterization of the problem that treats the parameters $(A, \Sg)$ equally.

Fortunately, the introduction of a fixed base density $\phi(\cdot)$ provides sufficient additional structure that enables such reparameterization. Note that, as dictated by the linear dynamics \eqref{eq:system_dynamics}, given a known base density $\phi(w)$ of the unit-covariance noise $w$, the state transition distribution $\P(x' | x)$ must also be associated with a conditional density $f_{A, \Sg}(x' | x)$ given by
\begin{equation}
  f_{A,\Sg}(x'|x) = \det(\Sg)^{-1/2} \phi\prn[\Big]{ \Sg^{-1/2} (x'-Ax) },
\end{equation}
which is parameterized by, and only by, $(A, \Sg)$. We say that $\P(x' | x)$ belongs to a \emph{$\phi$-density family}%
\footnote{This is closely related to the \textit{location-scale family} in probability theory literature, which refers to a collection of distributions/densities parameterized by location and scale parameters. In our case, $Ax$ and $\Sg$ can be viewed as the location and the scale, respectively. An expository discussion can be found in \cite{shao2008mathematical}.}. %
We will denote the the gradient and Hessian of $\log \phi(w)$ as $\g(w)$ and $\H(w)$, respectively, 
  \begin{equation}\label{eq:log_density_gradient-hessian}
    \g(w) = \nabla_{w} \log \phi(w); \quad 
    \H(w) = \nabla_{w}^2 \log \phi(w) 
  \end{equation}

For technical clarity, we further assume that the hypothesis space $\varTheta$ for $(A, \Sg)$ satisfies the following regularity assumption.

\begin{assumption}[Hypothesis space]\label{assum:regular_Theta}
  Let $\varTheta = \varTheta_A \times \varTheta_{\Sg}$ be the product of hypothesis spaces for $A$ and $\Sg$, where $\varTheta_A \subseteq \R^{d \times d}$ and $\varTheta_{\Sg} \subseteq \S^d_{\succ \zero}$. We assume that $\thetas = (\As, \Ss)$ lies in the interior of $\varTheta$, and that $f_{A_1, \Sg_1}(x' | x) \equiv f_{A_2, \Sg_2}(x' | x)$, $\forall x, x' \in \R^d$ if and only if $(A_1, \Sg_1) = (A_2, \Sg_2) \in \varTheta$.
\end{assumption}

The benefit of fixing a base density and considering the corresponding $\phi$-density family is that it provides a unified parameterization of transition distributions using $(A, \Sg)$ (rather than the distribution-specific parameters), directly exposing the parameters to learn and thus simplifying the design of estimators.

\boldtitle{The Elliptical Family.} Despite the generality of $\phi$-density families, it could be analytically and computationally hard to handle the log-densities and their derivatives in closed forms. To further facilitate the design and analysis of algorithms, we consider an important special type of $\phi$-density families, where their base densities are assumed to be \emph{elliptical} (i.e., the iso-density contours are ellipsoids; see \cite{cambanis1981theory} for details).

\begin{definition}[Elliptical density]
  A density $\phi$ is called \emph{elliptical} if there exists a function $\psi: \R \to \R$ such that $\phi(w) = \psi(w^{\top} w)$.
\end{definition}

Elliptical base densities are favorable since the derivatives of their log-densities can be written in simple closed forms.

\begin{lemma}\label{thm:lemma-diff_elliptical}
  For an elliptical base density $\phi(w) = \psi(w^{\top} w)$, the gradient $\g(w)$ and Hessian $\H(w)$ of $\log \phi(w)$ are given by
  \begin{subequations}\label{eq:elliptical-log_density_gradient}
  \begin{align}
    \g(w) &= \rho(w^{\top} w) w, \\
    \H(w) &= \rho(w^{\top} w) I + \tau(w^{\top} w) ww^{\top},
  \end{align}
  \end{subequations}
  where the functions $\rho: \R \to \R$ and $\tau: \R \to \R$ are specified by
  \begin{align*}
    \rho(z) &:= 2 \frac{\diff}{\diff z}\brak[\big]{ \log \psi(z) } = \frac{2\psi'(z)}{\psi(z)},\\
    \tau(z) &:= 4 \frac{\diff^2}{\diff z^2}\brak[\big]{ \log \psi(z) } = \frac{4 \prn[\big]{ \psi''(z) \psi(z) - \psi'(z)^2 }}{\psi^2(z)}.
  \end{align*}
\end{lemma}

\begin{proof}
  By direct calculation. See \ref{sec:apdx-elliptical_lemma} for details.
\end{proof}

We would like to highlight that a wide range of named distributions fall into the category of elliptical distributions. As a sanity check, we first show that the Gaussian family is elliptical.

\begin{example}[Multivariate Gaussian family]
  The log-density of the unit $d$-dimensional Gaussian distribution $\mathcal{N}(\zero, I_d)$ is
  \begin{equation*}
    \log \phi_{\mathcal{N}}(w) = -\frac{1}{2} w^{\top} w - \frac{d}{2} \log (2\pi),
  \end{equation*}
  which is a function of $w^{\top} w$ and is hence elliptical. It can be verified that the derivatives of $\log \phi_{\mathcal{N}}(w)$ are specified by \eqref{eq:elliptical-log_density_gradient}, with $\rho_{\mathcal{N}}$ and $\tau_{\mathcal{N}}$ given by
  \begin{equation*}
    \rho_{\mathcal{N}}(z) = -1,\qquad
    \tau_{\mathcal{N}}(z) = 0.
  \end{equation*}
\end{example}

We proceed to introduce another elliptical density that will be extensively used in simulations.

\begin{example}[Multivariate Student-t family]\label{eg:student-t}
  The log-density of the $d$-dimensional unit Student-t distribution $\mathsf{t}_{d,\nu}\prn[\big]{ \zero, \frac{\nu-2}{\nu} I_d }$ with degree-of-freedom parameter $\nu>2$ is given by
  \begin{equation*}
    \log \phi_{\mathsf{t}}(w) = - \frac{\nu+d}{2} \log\prn*{ 1 + \frac{1}{\nu-2} w^{\top} w },
  \end{equation*}
  which is a function of $w^{\top} w$ and is hence elliptical. It can be verified that the derivatives of $\log \phi_{\mathsf{t}}(w)$ are specified by \eqref{eq:elliptical-log_density_gradient}, with $\rho_{\mathsf{t}}$ and $\tau_{\mathsf{t}}$ given by
  \begin{equation*}
    \rho_{\mathsf{t}}(z) = -\frac{\nu+d}{(\nu-2) + z},\qquad
    \tau_{\mathsf{t}}(z) = \frac{2(\nu+d)}{ \prn{(\nu-2) + z}^2 }.
  \end{equation*}
\end{example}

%% file: 03-results.tex
\section{Joint Identification}\label{sec:3-algorithm}

In this section, we provide our solutions to the joint identification problem formulated above. We start by a motivating section that illustrates the issue with naive ordinary least-squares (OLS) estimators. Then we proceed to introduce two well-known estimators that solve the problem, namely the maximum likelihood estimator (MLE) and the score-matching estimator (SME), along with their characterizations and convergence properties in our setting. Finally, we present a sample complexity guarantee for the proposed estimators.

\subsection{The Naive OLS estimator and its limitations}

At first glance, one may naturally wonder how this problem could potentially be different from the standard system identification problem that only identifies the $A$ matrix (see, e.g., \cite{sarkar2019near}). Indeed, it is intuitive to come up with the following estimator $(\hAols, \hSols)$ that first identifies $A$ by solving a least squares problem, i.e.,
\begin{equation}\label{eq:estimator-OLS-A_optim}
  \hAols = \arg\min_{A} \sum_{i=1}^{N} \norm{x'_i - A x_i}_2^2,
\end{equation}
and then estimates $\Sg$ by the sample covariance calculated using $\hAols$. In closed-form formula, the \textit{naive OLS estimator} is
\begin{equation}\label{eq:estimator-OLS}
  \left\lbrace \begin{array}{@{}l}
    \hAols = \prn*{\sum_{i=1}^{N} x'_i x_i^{\top}} \prn*{\sum_{i=1}^{N} x_i x_i^{\top}}^{-1}, \\
    \hSols = \frac{1}{N} \sum_{i=1}^{N} \prn[\big]{ x'_i - \hAols x_i } \prn[\big]{ x'_i - \hAols x_i }^{\top}.
  \end{array} \right.
\end{equation}
However, a major concern of the above two-step procedure is that it completely ignores the information encoded in the $\phi$-density family of candidate distributions. To briefly illustrate this point, we consider the 1-dimensional special case%
\footnote{Note that both $A$ and $\Sg$ are scalars in this 1-dimensional case.} %
where the transition dataset $\D = \set{(x_i, x'_i)}$ is collected with current state fixed to $x_i \equiv 1$, such that the naive OLS estimators become
\begin{equation}\label{eq:estimator-OLS-1d_example}
  \left\lbrace \begin{array}{@{}l}
    \hAols = \frac{1}{N} \sum_{i=1}^{N} x'_i, \\
    \hSols = \frac{1}{N} \sum_{i=1}^{N} \prn[\big]{ x'_i - \hAols }^2.
  \end{array} \right.
\end{equation}
In other words, in this case $(A, \Sg)$ coincides with the location-scale parameter of the distribution of $x'$, while its naive OLS estimator becomes the sample mean and variance. However, it is known that, in the identification problem of the best location parameter with any given scale parameter, the maximum likelihood estimator (MLE) coincides with the sample mean if and only if the candidate family is Gaussian \cite{teicher1961maximum, marshall1993maximum}, which agrees with a similar result established later of this paper, where MLE only coincides with the OLS estimator when the base density is Gaussian (see \Cref{thm:equiv_estimators_Gaussian}). This result implies that the naive OLS estimator given in \eqref{eq:estimator-OLS-1d_example} is not necessarily the most probable guess we can make given the candidate family of distributions. An intuitive reason behind such discrepancy is that the OLS formulation fails to take into consideration the ``shape'' of the distribution for identification, and thus its performance may deteriorate when the family is vastly different from Gaussian. We emphasize that such observation motivates the introduction of alternative estimators that can better leverage the ``shape'' information of distributions encapsulated in the candidate family, as shown in the following sections.

\boldtitle{The Proposed Estimators.} For a brief overlook, we summarize the proposed OLS, MLE and SME estimators in \Cref{tab:estimator_overview} above, where the definition of losses $\hLmle$ and $\hLsme$ can be found in subsequent sections. We also include the information needed to calculate these estimators---OLS only requires the transition dataset, while MLE and SME require additional information like log-densities or their derivatives. 

In the following sections, we will proceed to introduce the proposed MLE and SME estimators in details.

\begin{table}
  \centering
  \begin{footnotesize}
  \begin{tabular}{c|c|c}
    \specialrule{1.0pt}{0pt}{0pt}
      \textbf{Estimator} & \textbf{Formula} & \textbf{Information} \\\hline
      OLS & $\begin{gathered}
        \displaystyle \hAols := \arg\min_{A} \textstyle\sum_{i=1}^{N} \norm{x'_i - A x_i}_2^2\\
        \hSols := \textstyle \frac{1}{N} \sum_{i=1}^{N} \prn[\big]{ x'_i - \hAols x_i } \prn[\big]{ x'_i - \hAols x_i }^{\top}
      \end{gathered}$ & $\D$ \\\hline
      MLE & $\displaystyle (\hAmle, \hSmle) := \arg\max_{A,\Sg} \hLmle(A,\Sg)$ & $\D$, $\log \phi(\cdot)$ \\\hline
      SME & $\displaystyle (\hAsme, \hSsme) := \arg\min_{A,\Sg} \hLsme(A,\Sg)$ & $\D$, $\g(\cdot)$, $\H(\cdot)$ \\\hline
    \specialrule{1.0pt}{0pt}{0pt}
  \end{tabular}
  \end{footnotesize}
  \caption{An overview of OLS, MLE and SME estimators.}\label{tab:estimator_overview}
\end{table}

\subsection{The Maximum Likelihood Estimator (MLE)}

We first define the maximum likelihood estimator (MLE) with respect to a $\phi$-density family $\set{f_{A,\Sg}}$ with base density $\phi(w)$.

\noindent\fbox{%
\parbox{\dimexpr\linewidth-2\fboxsep-2\fboxrule\relax}{%
\begin{align}\label{eq:MLE-definition}
  \textbf{(MLE)}\quad &(\hAmle, \hSmle) := \arg\max_{A,\Sg}\; \hLmle(A,\Sg),\\
  \text{where}\quad &\hLmle(A,\Sg) := \frac{1}{N} \sum_{i=1}^{N} \log f_{A,\Sg}(x'_i | x_i) \nonumber\\
  &\quad{}= \frac{1}{N} \sum_{i=1}^{N} \brak*{ \log \phi\prn[\Big]{ \Sg^{-1/2} (x'_i-Ax_i) } + \log\det \prn[\Big]{\Sg^{-1/2}} }. \nonumber
\end{align}
}}

\noindent We point out that the empirical loss $\hLmle$ could be understood as the sample-based estimator of the expected loss
\begin{equation*}
  \Lmle(A,\Sg) := \E[(x,x') \sim d^{\D}]{\log f_{A,\Sg}(x' | x)}.
\end{equation*}

Now we proceed to establish two fundamental properties of the MLE---(1) the unbiasedness of MLE, in the sense that the expected loss $\Lmle$ is uniquely maximized at the ground-truth parameters $(\As, \Ss)$; and (2) the characterization of the MLE, which promotes theoretical understanding of the algorithm.

\begin{lemma}[consistency of MLE]\label{thm:MLE-unbiased}
  Under \Cref{assum:regular_Theta} and \ref{assum:data-single_traj}, $(\As, \Ss)$ is the unique maximizer of $\Lmle(A, \Sg)$.
\end{lemma}

\begin{proof}
  Note that, by Gibbs' inequality we have
  \begin{align*}
    \Lmle(A, \Sg)
    &= \int_{x} \dD(x) \int_{x'} f_{\As,\Ss}(x' | x) \log f_{A,\Sg}(x' | x) \diff x' \diff x \\
    &\leq \int_{x} \dD(x) \int_{x'} f_{\As,\Ss}(x' | x) \log f_{\As,\Ss}(x' | x) \diff x' \diff x,
  \end{align*}
  where equality is achieved if and only if
  \begin{equation*}
    f_{A,\Sg}(\cdot | x) = f_{\As,\Ss}(\cdot | x),~ \forall x \in \R^d.
  \end{equation*}
  Then the proof is completed by applying \Cref{assum:regular_Theta}. 
\end{proof}

\rev{
\begin{lemma}[asymptotic normality of MLE]
  Under \Cref{assum:diff_base_density} and \ref{assum:regular_Theta}, let $\htheta := (\hA, \hS)$ and $\thetas := (\As, \Ss)$, such that $\sqrt{N} \prn[\big]{ \htheta - \thetas } \todist \mathcal{N}\prn*{0, \frac{1}{I(\theta)}}$, where $I(\theta) := \Var[\theta]{\nabla_{\theta} \log f_{\theta}(x)}$ denotes the Fisher information matrix.
\end{lemma}
}

\rev{
\begin{proof}
  This is a direct corollary of the asymptotic normality results for general MLEs (see, e.g., Theorem 9.18 in \cite{wasserman2013all}). In fact, it only suffices to verify that $f_{\theta}$ is differentiable in $\theta$ (\Cref{assum:diff_base_density}), $\thetas$ lies in the interior of $\varTheta$ (\Cref{assum:regular_Theta}), and $\thetas$ is the unique maximizer of $\Lmle(A, \Sg)$ (\Cref{thm:MLE-unbiased}).
\end{proof}
}

\begin{lemma}[MLE characterization]\label{thm:MLE-character_general}
  Under \Cref{assum:diff_base_density}, the MLE $(\hAmle, \hSmle)$ with respect to a $\phi$-density family $\set{f_{A,\Sg}}$ with base density $\phi$ satisfies the following system of equations:
  \begin{equation*}
    \left\lbrace \begin{array}{@{}l}
      \textstyle\sum_{i=1}^{N} \g\prn[\Big]{ \hSmle^{-1/2} \prn[\big]{ x'_i - \hAmle x_i } } x_i^{\top} = \zero,\\
      \hSmle^{1/2} + \frac{1}{N} \textstyle\sum_{i=1}^{N} \g\prn[\Big]{ \hSmle^{-1/2} \prn[\big]{ x'_i - \hAmle x_i } } \prn[\big]{ x'_i - \hAmle x_i }^{\top} = \zero.
    \end{array} \right.
  \end{equation*}
  Recall that $\g(w) := \nabla_{w} \log \phi(w)$ is the derivative of the log-base-density (see \Cref{assum:diff_base_density} for detailed definition).
\end{lemma}

\begin{proof}
  To apply the first-order criteria, we take the partial derivative of $\hLmle$ with respect to $A$ to obtain
  \begin{equation*}
    \zero = \left. \frac{\partial \hLmle}{\partial A} \right\vert_{\hAmle, \hSmle} = - \hSmle^{-1/2} \sum_{i=1}^{N} \g\prn[\Big]{ \hSmle^{-1/2} \prn[\big]{ x'_i - \hAmle x_i } } x_i^{\top},
  \end{equation*}
  which is equivalent to the first equation since $\hSmle \in \varTheta_{\Sg} \subseteq \S^d_{\succ \zero}$. Similarly, we take the partial derivative of $\hLmle$ with respect to $\Sg^{-1/2}$ (instead of $\varSigma$ for the sake of simplicity) to obtain
  \begin{align*}
    \zero &= \left. \frac{\partial \hLmle}{\partial \Sg^{-1/2}} \right\vert_{\hAmle, \hSmle} \\
    &= \frac{1}{N} \sum_{i=1}^{N} \brak*{ \g\prn[\Big]{ \hSmle^{-1/2} \prn[\big]{ x'_i - \hAmle x_i } } \prn[\big]{ x'_i - \hAmle x_i }^{\top} + \hSmle^{1/2} } \\
    &= \hSmle^{1/2} + \frac{1}{N} \textstyle\sum_{i=1}^{N} \g\prn[\Big]{ \hSmle^{-1/2} \prn[\big]{ x'_i - \hAmle x_i } } \prn[\big]{ x'_i - \hAmle x_i }^{\top},
  \end{align*}
  where we apply formulae \eqref{eq:matrix_diff-log_det} and \eqref{eq:matrix_diff-CAx}.
\end{proof}

The above characterization appears to be complicated for a generic $\phi$-density family. Nevertheless, for elliptical families the solution can be expressed in simplified forms that better connects to the OLS estimator, as shown in the lemma below.

\begin{corollary}[MLE characterization, elliptical family]\label{thm:MLE-character_elliptical}
  For an elliptical family $\set{f_{A,\Sg}}$ with base density $\phi(w) = \psi(w^{\top} w)$, under \Cref{assum:diff_base_density} and the regularity assumption that $\sum_{i=1}^{N} \rho_i x_i x_i^{\top} \succ \zero$, the MLE $(\hAmle, \hSmle)$ satisfies the following system of equations:
  \begin{equation}\label{eq:MLE-character_elliptical}
    \left\lbrace \begin{array}{@{}l}
      \hAmle = \prn*{\sum_{i=1}^{N} \rho_i x'_i x_i^{\top}} \prn*{\sum_{i=1}^{N} \rho_i x_i x_i^{\top}}^{-1},\\
      \hSmle = - \frac{1}{N} \sum_{i=1}^{N} \rho_i \prn[\big]{ x'_i - \hAmle x_i } \prn[\big]{ x'_i - \hAmle x_i }^{\top},
    \end{array} \right.
  \end{equation}
  where
  \begin{equation*}
    \rho_i = \rho\prn[\Big]{ \prn[\big]{ x'_i - \hAmle x_i }^{\top} \hSmle^{-1} \prn[\big]{ x'_i - \hAmle x_i } }
  \end{equation*}
  are dataset-specific coefficients determined by the base density, where $\rho(z) = \frac{2\psi'(z)}{\psi(z)}$ (see \Cref{thm:lemma-diff_elliptical} for details).
\end{corollary}

\begin{proof}
  Plugging \Cref{thm:lemma-diff_elliptical} into \Cref{thm:MLE-character_general} gives the result.
\end{proof}

The above corollary clearly reveals the connection of MLE and the OLS estimator for elliptical families. On the one hand, MLE is in general different from the OLS estimator when $\rho_i$'s are non-degenerative (for a degenerative example, consider the case for Gaussian families, where $\rho_i \equiv -1$ is constant and is thus canceled out). On the other hand, MLE for elliptical families take a similar form as OLS estimators, which can be viewed a reweighted version of the OLS formula. Indeed, since $|\rho(z)| \to 0$ as $z \to +\infty$, coefficients in \eqref{eq:MLE-character_elliptical} could be interpreted as reducing the weights of ``outlier'' noise samples that are significantly larger than the typical noise magnitude, which potentially enhances the convergence performance of the estimator.

\subsection{The Score-Matching Estimator (SME)}

In this section, we consider another type of estimators, i.e. the score-matching estimators (SME) first introduced in \cite{hyvarinen2005estimation}:
\begin{equation*}
  (\hAsme, \hSsme) := \arg\min_{A,\Sg}\; \hLsme(A,\Sg),
\end{equation*}
where $\hLsme$ is selected as \emph{an} empirical version of
\begin{equation*}
  \Lsme(A,\Sg) := \E[(x,x') \sim d^{\D}]{ \norm[\big]{ \nabla_{x'} \log f_{A,\Sg}(x' | x) - \nabla_{x'} \log f_{\As,\Ss}(x' | x) }_2^2 }.
\end{equation*}
We point out that the expected loss $\Lsme$ is intractable in its current form, since the ground-truth density $f_{\As,\Ss}(x' | x)$ is involved. However, it is shown in \cite{hyvarinen2005estimation} that the expected loss can be equivalently rewritten in a tractable form using the integration-by-part trick, as summarized by the following lemma.

\begin{lemma}
  Under \Cref{assum:diff_base_density} and \ref{assum:data-single_traj}, we have
  \begin{align*}
    \Lsme(A, \Sg) &:= \mathbb{E}_{(x,x') \sim \dD} \Bigl[ \norm[\big]{\nabla_{x'} \log f_{A,\Sg}(x' | x)}_2^2 \\
    &\hspace{7em} {}+ 2 \tr\prn*{ \nabla_{x'}^2 \log f_{A,\Sg}(x' | x) } \Bigr].
  \end{align*}
\end{lemma}

\begin{proof}
  See Appendix A of \cite{hyvarinen2005estimation} for details.
\end{proof}

Consequently, the SME with respect to a $\phi$-density family $\set{f_{A,\Sg}}$ with base density $\phi(w)$ can be defined as follows.

\noindent\fbox{%
\parbox{\dimexpr\linewidth-2\fboxsep-2\fboxrule\relax}{%
\begin{align}\label{eq:SME-definition}
  \textbf{(SME)}\quad &(\hAsme, \hSsme) := \arg\min_{A,\Sg}\; \hLsme(A,\Sg),\\
  \text{where}\quad &\hLsme(A, \Sg) := \sum_{i=1}^{N} \Bigl[ \norm[\big]{\nabla_{x'} \log f_{A,\Sg}(x'_i | x_i)}_2^2 \nonumber\\
  &\hspace{7em} {}+ 2 \tr\prn*{ \nabla_{x'}^2 \log f_{A,\Sg}(x'_i | x_i) } \Bigr] \nonumber\\
  &\hspace{4.5em}{}= \sum_{i=1}^{N} \Bigl[ \norm[\big]{\Sg^{-1/2} \g\prn[\Big]{\Sg^{-1/2} (x'_i - Ax_i)} }_2^2 \nonumber\\
  &\hspace{7em} {}+ 2 \tr\prn*{ \Sg^{-1} \H\prn[\Big]{\Sg^{-1/2} (x'_i - Ax_i)} } \Bigr]. \nonumber
\end{align}
}}

\noindent Recall that $\g(\cdot)$ and $\H(\cdot)$ are defined in \Cref{assum:diff_base_density}.

Parallel to the previous section, we first show that the SME is unbiased, as summarized in the following lemma.

\begin{lemma}[consistency of SME]\label{thm:SME-unbiased}
  Under \Cref{assum:regular_Theta}, $(\As, \Ss)$ is the unique maximizer of $\Lsme(A, \Sg)$.
\end{lemma}

\begin{proof}
  This result is evident from the original definition of $\Lsme$. On the one hand, $\Lsme$ is non-negative, with the minimum $0$ achievable at $(A, \Sg) = (\As, \Ss)$. On the other hand, for any $\nabla_{x'} \log f_{A,\Sg}(x' | x) = \nabla_{x'} \log f_{\As,\Ss}(x' | x)$, we must have $\log f_{A,\Sg}(x' | x) = \log f_{\As,\Ss}(x' | x) + C$, where $C = 0$ due to the unit normalization of densities. Finally, the proof is completed by applying \Cref{assum:regular_Theta}.
\end{proof}

\rev{
Parallel to the previous section, we can also establish asymptotic normality of SME as a corollary of the generic results (see, e.g., Corollary 1 in \cite{song2020sliced}). In fact, under \Cref{assum:diff_base_density}, \ref{assum:regular_Theta} and some additional regularity assumptions of the score function, we can show that $\sqrt{N} \prn[\big]{ \htheta - \thetas } \todist \mathcal{N}\prn*{0, \varXi}$, where $\varXi := \prn[\big]{ \nabla_{\theta}^2 \mathcal{L}_{\mathrm{sme}}(\thetas) }^{-1} \prn{\sum_{i,j} V_{i,j}} \prn[\big]{ \nabla_{\theta}^2 \mathcal{L}_{\mathrm{sme}}(\thetas) }^{-1}$ is an instance-specific covariance matrix specified therein. However, the proof is purely technical and is thus omitted here.
}

One may expect a similar characterization lemma for SME. However, due to the complexity in its formulation, it is impractical to write down the equations. Fortunately, for elliptical families, we can still apply \Cref{thm:lemma-diff_elliptical} to obtain a simple and neat characterization for SME parallel to \Cref{thm:MLE-character_elliptical}.

\begin{lemma}[SME characterization, elliptical family]\label{thm:SME-character_elliptical}
  For an elliptical family $\set{f_{A,\Sg}}$ with base density $\phi(w) = \psi(w^{\top} w)$, under \Cref{assum:diff_base_density} and the regularity assumption $\sum_{i=1}^{N} (\rho_i^2 + 2\tau_i) x_i x_i^{\top} \succ \zero$, the SME $(\hAsme, \hSsme)$ satisfies the following system of equations:
  \begin{equation}\label{eq:SME-character_elliptical}
    \left\lbrace \begin{array}{@{}l}
      \hAsme = \prn*{\sum_{i=1}^{N} (\rho_i^2 + 2\tau_i) x'_i x_i^{\top}} \prn*{\sum_{i=1}^{N} (\rho_i^2 + 2\tau_i) x_i x_i^{\top}}^{-1},\\
      \hSsme = - \frac{1}{N} \sum_{i=1}^{N} \frac{\rho_i^2 + 2\tau_i}{\rho_i} \prn[\big]{ x'_i - \hAsme x_i } \prn[\big]{ x'_i - \hAsme x_i }^{\top},
    \end{array} \right.
  \end{equation}
  where
  \begin{align*}
    \rho_i &= \rho\prn[\Big]{ \prn[\big]{ x'_i - \hAsme x_i }^{\top} \hSsme^{-1} \prn[\big]{ x'_i - \hAsme x_i } },\\
    \tau_i &= \tau\prn[\Big]{ \prn[\big]{ x'_i - \hAsme x_i }^{\top} \hSsme^{-1} \prn[\big]{ x'_i - \hAsme x_i } }
  \end{align*}
  are dataset-specific coefficients determined by the base density, where $\rho(z) = \frac{2\psi'(z)}{\psi(z)}$ and $\tau(z) = \frac{4 \prn{ \psi''(z) \psi(z) - \psi'(z)^2 }}{\psi^2(z)}$ (see \Cref{thm:lemma-diff_elliptical}).
\end{lemma}

\begin{proof}
  Plug \Cref{thm:lemma-diff_elliptical} into $\hLsme(A, \Sg)$, and we have
  \begin{align*}
    \hLsme(A, \Sg) &= \sum_{i=1}^{N} \rho_i^2 \norm[\big]{\Sg^{-1} (x'_i - Ax_i)}_2^2 + 2\rho_i \tr(\Sg^{-1}) \\
      &\hspace{2em}{}+ 2\tau_i \tr\prn[\Big]{ \Sg^{-3/2} (x'_i-Ax_i) (x'_i - Ax_i)^{\top} \Sg^{-1/2} } \\
    &= \sum_{i=1}^{N} (\rho_i^2 + 2\tau_i) \norm[\big]{\Sg^{-1} (x'_i - Ax_i)}_2^2 + 2\rho_i \tr(\Sg^{-1}).
  \end{align*}
  Therefore, to apply the first-order criteria, we take the partial derivative of $\hLsme$ with respect to $A$ to obtain
  \begin{equation*}
    \zero = \left. \frac{\partial \hLsme}{\partial A} \right\vert_{\hAsme , \hSsme} = 2 \hSsme^{-2} \sum_{i=1}^{N} (\rho_i^2+2\tau_i) \prn[\big]{ x'_i - \hAsme x_i } x_i^{\top},
  \end{equation*}
  which is equivalent to the first equation since $\hSsme \in \varTheta_{\Sg} \subseteq \S^d_{\succ \zero}$. Similarly, we take the partial derivative of $\hLsme$ with respect to $\Sg^{-1}$ (instead of $\varSigma$ for the sake of simplicity) to obtain
  \begin{align*}
    \zero &= \left. \frac{\partial \hLsme}{\partial \Sg^{-1}} \right\vert_{\hAsme, \hSsme} \\
    &= \sum_{i=1}^{N} 2(\rho_i^2 + 2\tau_i) \hSsme^{-1} \prn[\big]{ x'_i - \hAsme x_i } \prn[\big]{ x'_i - \hAsme x_i }^{\top} + 2\rho_i I,
  \end{align*}
  which is equivalent to the second equation. Note that here we apply formulae \eqref{eq:matrix_diff-log_det}, \eqref{eq:matrix_diff-tr_A} and \eqref{eq:matrix_diff-Ax+b}.
\end{proof}

We point out that the form of \eqref{eq:SME-character_elliptical} is very similar to that of \eqref{eq:MLE-character_elliptical}, both regarded as reweighted versions of the OLS formula.

\subsection{Unified Estimators for Gaussian Noise}\label{sec:3-4-1-Gaussian}

As briefly discussed above, it is intuitive to expect that the MLE and SME estimators coincide with the naive OLS estimator for the Gaussian family, which is indeed the case and can be easily verified as a sanity check in the following lemma.

\begin{lemma}[Unified Gaussian estimators]\label{thm:equiv_estimators_Gaussian}
  For the Gaussian family (i.e., one with base density $\phi(x) = \frac{1}{\sqrt{2\pi}} \exp\prn[\big]{ -\frac{x^2}{2} }$), under the regularity assumption that $\sum_{i=1}^{N} x_i x_i^{\top} \succ \zero$, the MLE and SME estimators both coincide with the naive OLS estimator, i.e.
  \begin{equation*}
    (\hAmle, \hSmle) = (\hAsme, \hSsme) = (\hAols, \hSols).
  \end{equation*}
\end{lemma}

\begin{proof}
  Note that $\rho(z) \equiv -1$ and $\tau(z) \equiv 0$, and the results follows from \Cref{thm:MLE-character_elliptical} and \Cref{thm:SME-character_elliptical} immediately.
\end{proof}

The above lemma also serves as a piece of supporting evidence that the MLE and SME estimators are supposed to outperform the naive OLS estimator by leveraging the shape information implicitly encoded in the $\phi$-density family. On the contrary, since Gaussian distribution is the max-entropy distribution of prescribed mean and covariance, the estimators naturally degenerate into the naive OLS estimators due to a lack of additional structural information.

\subsection{Numerical Implementations}\label{sec:3-4-2-numerical}

In this section, we briefly discuss how to numerically compute the proposed OLS, MLE and SME estimators efficiently.

\boldtitle{OLS.} The OLS estimator can be computed either using the closed-form formula in \eqref{eq:estimator-OLS}, or by solving the optimization in \eqref{eq:estimator-OLS-A_optim} using a generic gradient-based solver, where the latter may be preferred for large datasets to avoid numerical issues.

\boldtitle{MLE and SME.} In general, the MLE and SME estimators should be computed by solving the optimization formulated in \eqref{eq:MLE-definition} and \eqref{eq:SME-definition}, respectively, using a generic gradient-based solver.

Alternatively, since MLE and SME for elliptical families are known to be reweighted OLS estimators in the form
\begin{equation*}
  \left\lbrace \begin{array}{@{}l}
    \hA = \prn*{\sum_{i=1}^{N} \lambda_i(\hA, \hS) x'_i x_i^{\top}} \prn*{\sum_{i=1}^{N} \lambda_i(\hA, \hS) x_i x_i^{\top}}^{-1},\\
    \hS = - \frac{1}{N} \sum_{i=1}^{N} \kappa_i(\hA, \hS) \prn[\big]{ x'_i - \hA x_i } \prn[\big]{ x'_i - \hA x_i }^{\top},
  \end{array} \right.
\end{equation*}
with weight functions $\lambda_i(\cdot)$ and $\kappa_i(\cdot)$ determined by data $(x_i, x'_i)$, we may apply the following iterative algorithm that produces a sequence $\set{(\hA_k, \hS_k) \mid k=1,2,\ldots}$, where
\begin{equation*}
  \left\lbrace \begin{array}{@{}l}
    \hA_{k+1} = \prn*{\sum_{i=1}^{N} \lambda_i(\hA_k, \hS_k) x'_i x_i^{\top}} \prn*{\sum_{i=1}^{N} \lambda_i(\hA_k, \hS_k) x_i x_i^{\top}}^{-1},\\
    \hS_{k+1}  = - \frac{1}{N} \sum_{i=1}^{N} \kappa_i(\hA_k, \hS_k) \prn[\big]{ x'_i - \hA x_i } \prn[\big]{ x'_i - \hA x_i }^{\top}.
  \end{array} \right.
\end{equation*}
It can be verified that the above iterative algorithm converges to the optimum for certain estimators and certain distributions (for example, MLE for the multivariate Student-t distribution, as shown in \Cref{sec:4-simulations} below), but a general convergence guarantee is hard to establish. Curious readers are referred to Chapter 13 of \cite{bilodeau1999theory} for a set of sufficient conditions that guarantee uniqueness and convergence of the iterative algorithm in the case without dynamics (i.e., the location parameter is simply $\mu$ instead of $Ax_i$). We leave the detailed analysis of this algorithm as future work.

\subsection{Single-trajectory Sample Complexity}

In this section, we present a single-trajectory sample complexity analysis for the proposed MLE and SME estimators in the case where the system is stable and the noise is sub-Gaussian.

\begin{assumption}[Single-trajectory data]\label{assum:data-single_traj}
  $\D$ contains transitions collected from a single trajectory; i.e., we have $\D = \set{(x_i, x_{i+1}) \mid t \in [N]}$ for a trajectory $x_0, x_1, x_2, \ldots$ subject to the dynamics $x_i = \As x_{i-1} + \Ss^{1/2} w_{i-1}$, $\forall i \in \N$.
\end{assumption}

\begin{assumption}[Sub-Gaussian noise]\label{assum:sub_gaussian_noise}
  The noise $w_i$ is i.i.d. sub-Gaussian with \emph{variance proxy} $\sigma^2$ (see \Cref{def:sub_gaussian}).
\end{assumption}

\begin{assumption}[Stability]\label{assum:stability}
  The system is stable in the sense that $\rho(\As) < 1$, and that the Gramian $\varGamma_t(A) := \sum_{\tau=0}^{t-1} A^{\tau} (A^{\tau})^{\top}$ is bounded by $\varGamma$. Consequently, $\norm{x_i} \leq C_x$, $\forall i \in [N]$.
\end{assumption}

\begin{assumption}[Bounded $\varTheta$]\label{assum:bounded_Theta}
  The hypothesis space $\varTheta$ is bounded such that $\norm{A} \leq C_A$, $\forall A \in \varTheta_A$ and $\norm{\Sg^{-1/2}} \leq C_{\Sg}$, $\forall \Sg \in \varTheta_{\Sg}$.
\end{assumption}

\begin{theorem}[Sample Complexity]\label{thm:main-sample_complexity}
  Under Assumptions \ref{assum:diff_base_density}, \ref{assum:regular_Theta}, \ref{assum:data-single_traj}, \ref{assum:sub_gaussian_noise}, \ref{assum:stability} and \ref{assum:bounded_Theta}, when $(\hA, \hS) \in \set[\big]{(\hAmle, \hSmle), (\hAsme, \hSsme)}$, we have with probability at least $1-\delta$
  \begin{align*}
    \norm*{\hA - \As} &\lesssim \sqrt{\frac{d \log (d/\delta)}{N}},\\
    \norm*{\hS - \Ss} &\lesssim \sqrt{\frac{d \log (dN/\delta) \log(1/\delta)}{N}}.
  \end{align*}
  Here the asymptotics hide instance-specific constants that do not depend on $d$, $N$ or $\delta$.
\end{theorem}

\begin{proof}[Proof sketch.]
  Given the reweighted OLS characterization for MLE and SME estimators (see \Cref{thm:MLE-character_elliptical} and \ref{thm:SME-character_elliptical}), we shall first apply known self-normalized martingale concentration bounds (see \cite{abbasi2011improved}) to derive the convergence rate of $A$. To further derive the convergence rate of $\Sg$, we may resort to the concentration of $\hL$ around $\L$, as well as the smoothness and strong convexity of the objective function over a bounded hypothesis space.
  
  Detailed proofs can be found in \ref{sec:apdx-sample_complexity}.
\end{proof}

The sample complexity guarantee shows that the proposed estimators achieves a convergence rate of $\tilde{\mathrm{O}}(N^{-1/2})$, which is identical to standard MLE estimators.

%% file: 04-simulations.tex
\section{Simulations}\label{sec:4-simulations}

In this section, we examine the numerical performance of the proposed MLE and SME estimators via simulation. We start by comparing the identification error under different sample sizes to provide empirical evidence for the sample complexity results. Then we compare the identification error under different system dimensions to showcase the scalability of the algorithms. We also evaluate the proposed estimators in terms of computation time to reflect their computational complexity, highlighting a trade-off between sample and computational complexity.

\boldtitle{Setting.} We consider the \emph{Student-t family} with base density $\phi_{\mathsf{t}}(w)$ selected as the multivariate Student-t distribution, as specified in \Cref{eg:student-t}. For simulations in a 2-d system with different sample sizes, the following system parameters are used:
\begin{equation*}
  \As = \begin{bmatrix}
    1 & 2 \\
    & \frac{1}{2}
  \end{bmatrix},\qquad
  \Ss = \begin{bmatrix}
    1 \\
    & 4
  \end{bmatrix}.
\end{equation*}
For simulations on scalability with different system dimensions, the dynamical matrix $\As$ is a tri-diagonal matrix, while $\Ss$ is simply selected as identity, i.e.
\begin{equation*}
  \As = \begin{bmatrix}
    1 & 1 \\
    1 & 1 & 1 \\
    & 1 & 1 & 1 \\
    & & \sddots & \sddots & \sddots
  \end{bmatrix}_{d \times d},\qquad
  \Ss = I_d.
\end{equation*}
All optimizations are solved by the general-purpose nonlinear solver provided by SciPy \cite{virtanen2020scipy} using the BFGS algorithm \cite{fletcher2000practical}. Experiments are repeated for 100 times to compute the median and confidence interval (marked with shadow).

\boldtitle{Sample Complexity.} To empirically evaluate the sample complexity of the MLE, SME and naive OLS estimators, we first compare their identification error in the 2-d system specified above. It is shown in \Cref{fig:error-sample_size} below that the estimation errors of all estimators converge to 0 as the number of samples grow, yet at different rates and displaying different patterns. When examining the curves closely, it can be observed that MLE has the smallest identification error in all settings, both in terms of the median and the confidence interval; further, while the estimation error of $A$ convergence at similar rates for all the three estimators, their performance for identifying $\Sg$ differs significantly, with OLS the worst, MLE the best, and SME in between. These observations match well with our initial motivation that better leveraging the ``shape'' information of the distribution family helps to improve the identification performance.

\begin{figure}[t]
    \centering
    \includegraphics[width=0.98\linewidth]{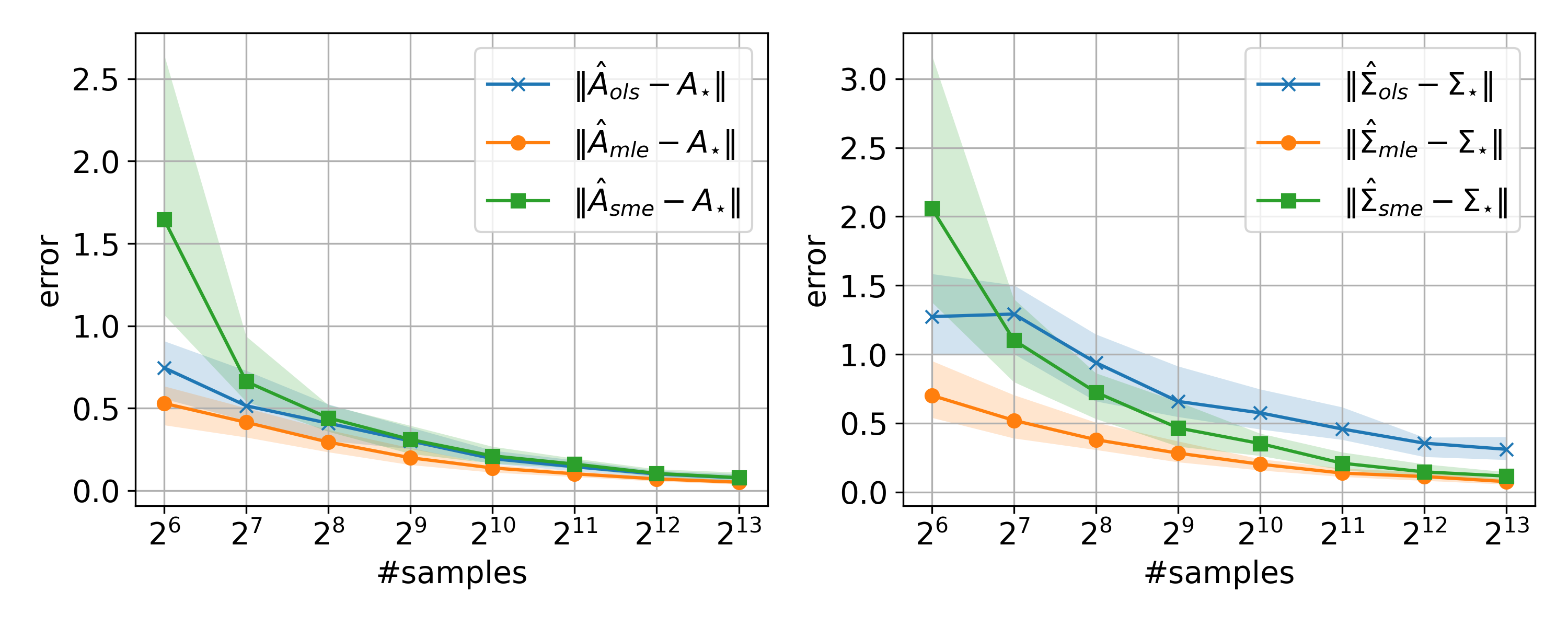}
    \caption{Identification error of MLE, SME and OLS estimators for different sample sizes (2-d Student-t family, log scale, shadowed confidence interval).}
    \label{fig:error-sample_size}
\end{figure}

\boldtitle{Scalability.} To reveal the scalability of the proposed estimators, we also measure the estimation error of the MLE and OLS estimators with different system dimensions using a fixed number of samples (256 in implementation), with SME excluded since the solver fails to output solutions for higher dimensions within reasonable time. As shown in \Cref{fig:error-dimension} below, it can be observed that the advantage of MLE over the naive OLS estimator becomes more significant as the system dimension, which can in turn be attributed to the increased dissimilarity of the Student-t family from the Gaussian family in higher dimensions.

\begin{figure}[t]
    \centering
    \includegraphics[width=0.98\linewidth]{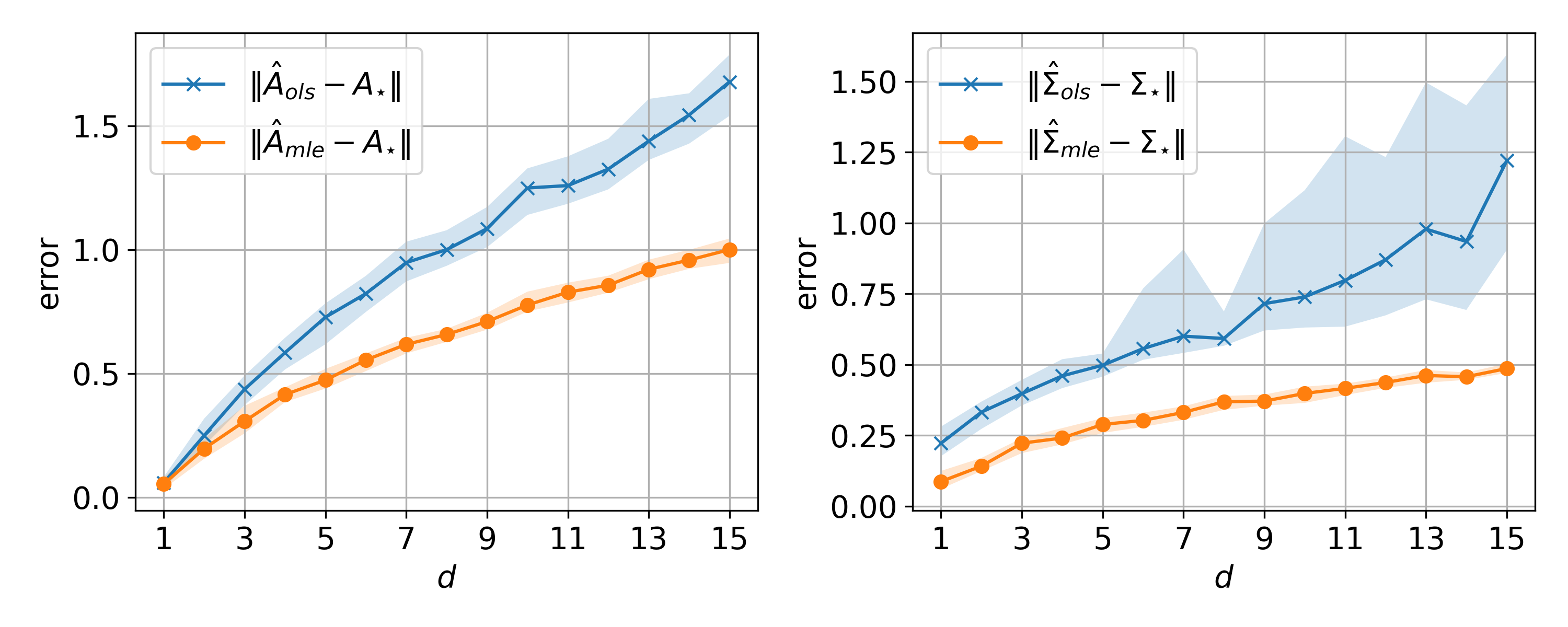}
    \caption{Identification error of MLE and OLS estimators for different dimensions (256 samples, log scale, shadowed confidence interval).}
     \label{fig:error-dimension}
\end{figure}

\boldtitle{Computational Efficiency.} We proceed to evaluate the computational efficiency of the proposed estimators under different sample sizes and different system dimensions. For this purpose, we include two different implementations of the MLE estimator---the standard implementation by solving the optimization, and the iterative implementation (labeled as ``\textsf{iter}'') by applying the iterative algorithm proposed in \Cref{sec:3-4-2-numerical}. It is shown in \Cref{fig:error-sample_size_iterMLE} that the iterative algorithm is able to approximate the MLE to a very high precision after only 20 iterations.

The average computation time are then measured and shown in \Cref{fig:wall_time}. It can be observed in \Cref{fig:wall_time-sample_size} that the computation of MLE and SME is significantly less efficient than the naive OLS estimator, where the latter comes with a closed-form formula and thus does not involve any optimization. Further, the computation time for MLE is reduced by approximately an order of magnitude when the iterative algorithm is applied, justifying the efficiency of the iterative algorithm. A similar slowdown of MLE computation (as compared to the naive OLS baseline) is observed in \Cref{fig:wall_time-dimension}, which deteriorates with higher system dimensions. We point out again that SME is not included in \Cref{fig:wall_time-dimension} because the solver simply fails to minimize $\hLsme$ in higher dimensions, probably due to the highly non-convex landscape of the optimization problem. A comparison of \Cref{fig:error-dimension} and \Cref{fig:wall_time-dimension} reveals a clear trade-off between sample complexity and computational efficiency.

\begin{figure}[t]
    \centering
    \includegraphics[width=0.98\linewidth]{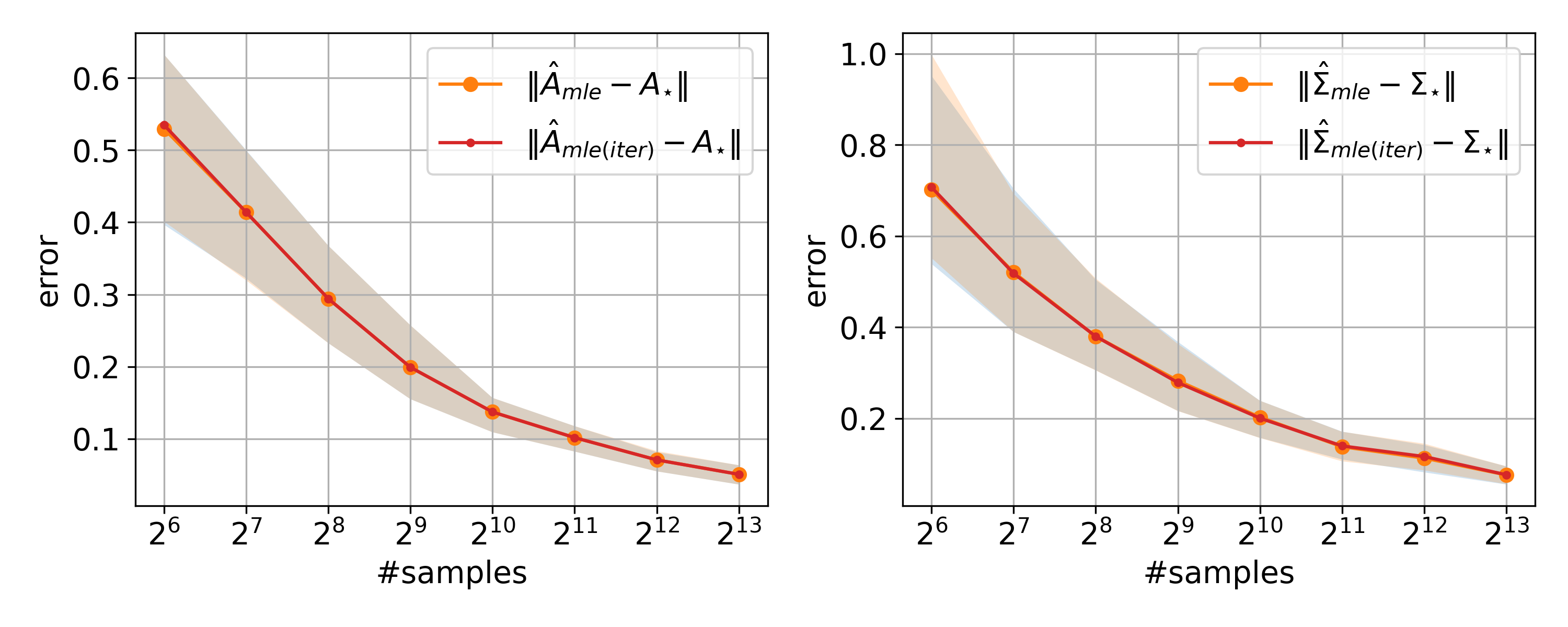}
    \caption{Identification error of different implementations of MLE for different sample sizes (2-d Student-t family, log scale, shadowed confidence interval).}
     \label{fig:error-sample_size_iterMLE}
\end{figure}

\begin{figure}[t]
    \centering
    \begin{subfigure}[b]{0.49\linewidth}
      \centering
      \includegraphics[width=0.98\linewidth]{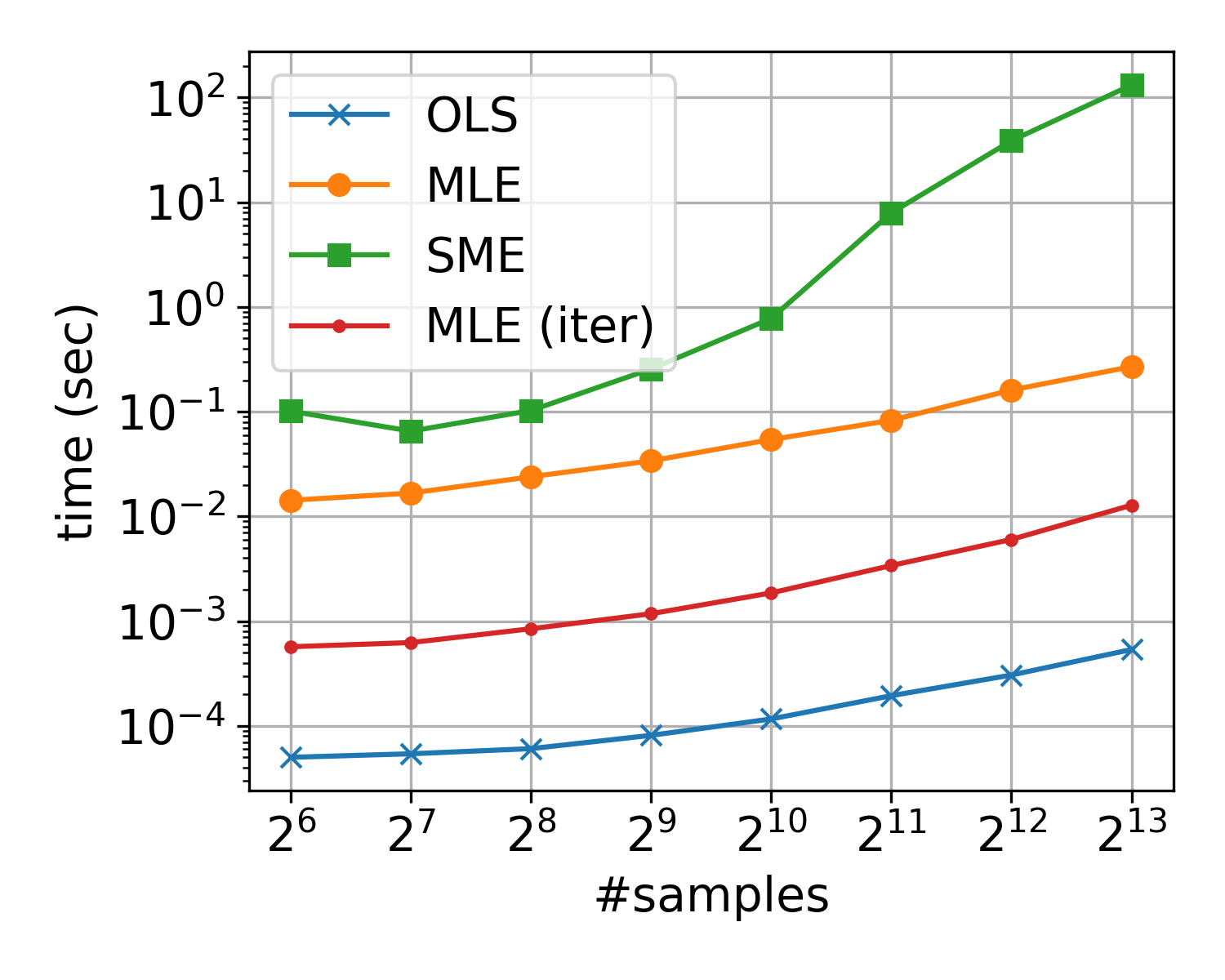}
      \caption{2-d system, different sample sizes}
      \label{fig:wall_time-sample_size}
    \end{subfigure}
    \begin{subfigure}[b]{0.49\linewidth}
      \centering
      \includegraphics[width=0.98\linewidth]{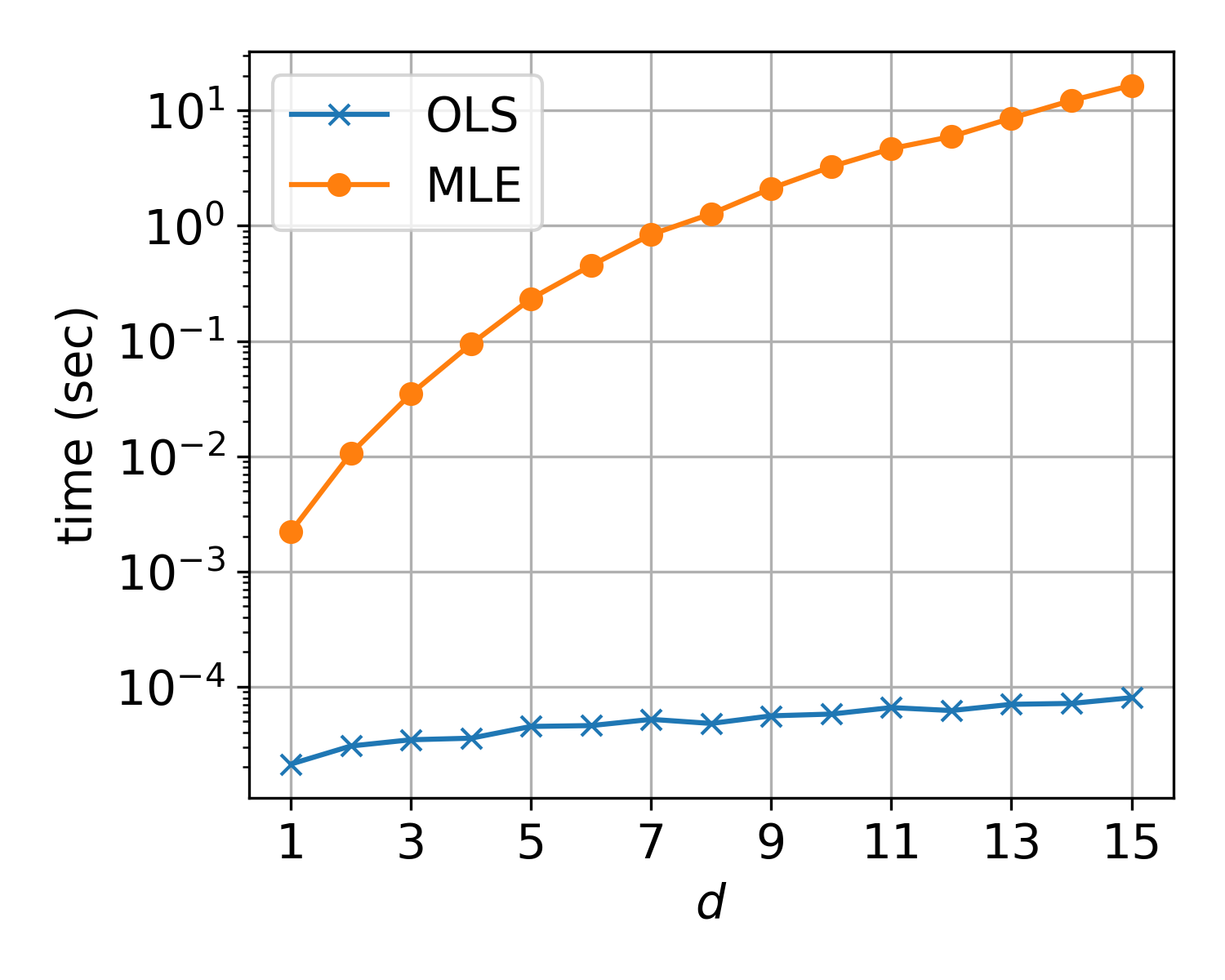}
      \caption{256 samples, different system dimensions}
      \label{fig:wall_time-dimension}
    \end{subfigure}
    \caption{Average computation time of MLE, SME and OLS estimators for different dimensions (Student-t family, log scale).}
    \label{fig:wall_time}
\end{figure}

\rev{
\boldtitle{Mis-specified Base Density.} Finally, we empirically illustrate the proposed estimators' robustness with respect to mis-specified base densities. Consider the case where the student-t density is perturbed, such that the $w^{\top} w$ term is replaced by $(w^{\top} w)^{1+\varepsilon}$, i.e.
\begin{equation*}
  \log \tilde{\phi}_{\mathsf{t}}^{\varepsilon} (w) = - \frac{\nu+d}{2 Z_{\varepsilon}} \log\left( 1 + \frac{1}{\nu-2} (w^{\top} w)^{1+\varepsilon} \right),
\end{equation*}
where $Z_{\varepsilon}$ is an appropriate normalizing factor for error level $\varepsilon$. The estimation errors at different error levels are plotted in the following \Cref{fig:perturbation}. It can be observed that $\hAmle$ slightly outperforms $\hAols$ in a consistent manner, while the prediction error of $\hSmle$ gradually deteriorates with larger $\varepsilon$, and eventually becomes worse than $\hSols$.
}

\rev{
Such deterioration in performance is quite intuitive, and also matches with our theoretical analysis. In fact, when the error $\norm{\tilde{\phi}(w) - \phi(w)}$ is supposed to be bounded, following the derivations presented in \ref{sec:apdx-sample_complexity}, we will see an additional drift term in the prediction error bound that positively correlates to the error level $\varepsilon$. The above discussion shows that the proposed estimators exhibit certain level of robustness to mis-specified base densities when the error level is mild.
}
    
\begin{figure}[h]
  \centering
  \includegraphics[width=0.98\linewidth]{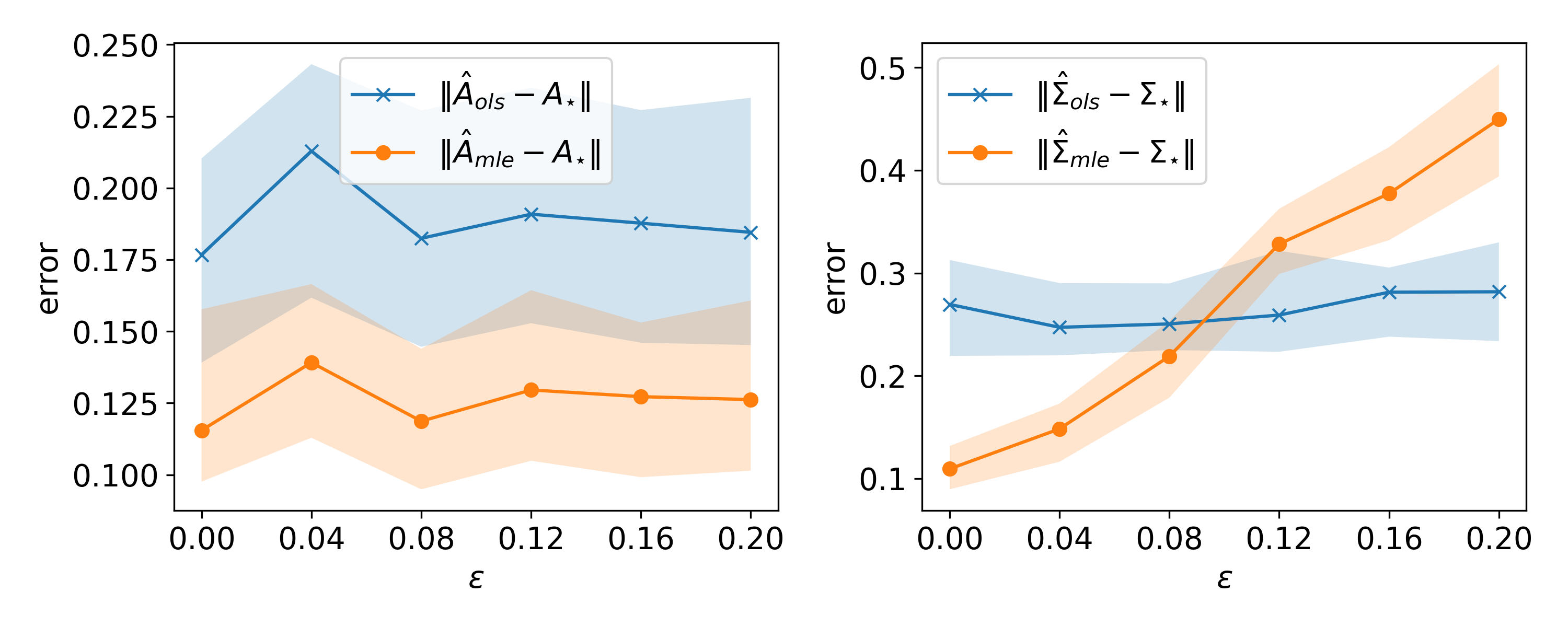}
  \caption{Identification error of MLE estimators for mis-specified base densities at different error levels (2-d perturbed Student-t family, log scale, shadowed confidence interval).}\label{fig:perturbation}
\end{figure}

%% file: 05-backmatters.tex
\section{Conclusion}

In this paper, we study sample-efficient joint identification of linear dynamics and noise covariance for generic noise distributions via a novel parameterization of the state-transition distribution. Building upon both theoretical and experimental evidence, we highlight the necessity to leverage the ``shape'' information of distributions that is implicitly encoded in the $\phi$-density families, and propose two practical estimators (i.e., MLE and SME) to solve the joint identification problem, along with characterization of their key properties and a sample complexity guarantee. Simulation results confirms the outperformance of the proposed estimators against the naive OLS baseline at higher computational costs, showing a trade-off between sample complexity and computational efficiency. Future work includes discovering better estimators to account for a wider range of distribution families and designing more efficient algorithms to compute the estimators.

\section*{Acknowledgement}

This paper is supported by the following grants: NSF ECCS 2328241 (ASCENT), NSF CBET 2112085 (AI Institute), NSF ECCS 2401390, ONR N000142512173.

%% file: A1-appendix.tex
\medskip
\boldtitle{Notations.} Let $\S^n$ denote the set of all $n$-by-$n$ symmetric matrices, where subscripts indicate positive-(semi)definiteness. Let $A \succ B$ ($A \succeq B$) be a shorthand for denoting $A-B$ to be positive (semi)definite. Let $\norm{\cdot}$ denote the Euclidean 2-norm for vectors and the induced 2-norm for matrices. Let $\E{X}$ and $\Var{X}$ denote the expectation and the covariance matrix of a vector-valued random variable $X$, respectively. Let $\det(A)$ denote the determinant of $A$. Let $\nabla_x f(x,y)$ denote the partial derivative of $f$ with respect to $x$, where $x$ and $y$ are allowed to be vectors.

\section{Proof of \Cref{thm:lemma-diff_elliptical}}\label{sec:apdx-elliptical_lemma}

\begin{proof}[Proof of \Cref{thm:lemma-diff_elliptical}]
  By direct computation we have
  \begin{align*}
    \g(w) &:= \nabla_{w} \log \phi(w) = \nabla_{w} \log \psi(w^{\top} w) \\
    &= 2 \left. \frac{\diff}{\diff z}[\log \psi(z)] \right\vert_{z = w^{\top} w} w \\
    &= \rho(w^{\top} w) w,
  \end{align*}
  where we apply \eqref{eq:matrix_diff-xTx} and the definition of $\rho(z)$. Similarly,
  \begin{align*}
    \H(w) &:= \nabla_{w}^2 \log \phi(w) \\
    &= 2 \left. \frac{\diff}{\diff z}[\log \psi(z)] \right\vert_{z = w^{\top} w} \cdot I + 4 \left. \frac{\diff^2}{\diff z^2}[\log \psi(z)] \right\vert_{z = w^{\top} w} \cdot ww^{\top} \\
    &= \rho(w^{\top} w) I + \tau(w^{\top} w) ww^{\top},
  \end{align*}
  where we apply \eqref{eq:matrix_diff-xTx_2} and the definitions of $\rho(z)$ and $\tau(z)$.
\end{proof}

\section{Sample Complexity Analysis}\label{sec:apdx-sample_complexity}

In this section, we present our sample complexity analysis in a generic way. Specifically, consider any \emph{unbiased} expected loss $\L(A, \Sg) = \E{\ell\prn[\big]{ \Sg^{-1/2}(x' - Ax), \Sg }}$, i.e.,
\begin{equation*}
  (\As, \Ss) = \arg\min_{A,\Sg} \L(A, \Sg).
\end{equation*}
The corresponding estimator for $(A, \Sg)$ is constructed by minimizing its empirical version, namely
\begin{align*}
  (\hA, \hS) &:= \arg\min_{A,\Sg} \hL(A, \Sg), \\
  \text{where}\quad \hL(A, \Sg) &:= \frac{1}{N} \sum_{i=1}^{N} \ell\prn[\big]{ \Sg^{-1/2}(x'_i - Ax_i), \Sg }.
\end{align*}
In this section, we make the following assumption on $\ell$.

\begin{assumption}[Well-conditioned losses]\label{assum:apdx-loss_condition}
  $\ell(w, \Sg)$ is $L$-smooth in $w$, and $\mu$-strongly convex in both $w$ and $\Sg$; $\ell(0, \Sg) = 0$, $\forall \Sg \in \varTheta_{\Sg}$, and $\E[w \sim \phi]{\ell(w, \Sg)} \leq C_{\ell}$.
\end{assumption}

In addition, prompted by \Cref{thm:MLE-character_elliptical} and \Cref{thm:SME-character_elliptical}, we also assume that the estimator $\hA$ can be expressed in a reweighted OLS form, as summarized in the following assumption.

\begin{assumption}[Reweighted OLS form]\label{assum:apdx-reweighted_OLS}
There exist weights $\lambda_i \in (\lolambda, \hilambda)$ such that $\hA = \prn*{\sum_{i=1}^{N} \lambda_i x'_i x_i^{\top}} \prn*{\sum_{i=1}^{N} \lambda_i x_i x_i^{\top}}^{-1}$.
\end{assumption}

\rev{We point out that assuming bounded instance-specific weights $\lambda_i$'s should not be deemed as restrictive in the following sense. In practice, we usually have an initial guess of $(A, \Sg)$, especially when we have some prior knowledge of the system. When the hypothesis space $\varTheta$ is selected as a sufficiently neighborhood around $(\As, \Ss)$, it can be shown that such bounds $\hilambda > \lolambda > 0$ exist with high probability, owing to the continuity of the solution in $A$ and $\Sg$.}

The following lemma specifies the sample complexity bound for estimating $A$ using data collected from a single trajectory.

\begin{lemma}\label{thm:sample_complexity_A}
  Under Assumptions \ref{assum:data-single_traj}, \ref{assum:stability}, \ref{assum:apdx-loss_condition} and \ref{assum:apdx-reweighted_OLS}, given any $\delta \in (0, \frac{1}{2})$, there exists an $N_0 \in \N$ where $N_0 = \tilde{\mathrm{O}}\prn*{ \prn[\big]{ \norm{\Ss^{1/2}} \sigma }^2 d \log \tfrac{1}{\delta} }$, such that for any $N > N_0$, with probability at least $1-\delta$,
  \begin{equation*}
    \norm*{\hA - \As} \lesssim \sqrt{\frac{d \log (d/\delta)}{N}}.
  \end{equation*}
\end{lemma}

\begin{proof}
  It is straight-forward to bound the error as follows:
  \begin{subequations}\label{eq:OLS_like_concentration:e1}
  \begin{align}
    \norm*{\hA - \As}
    &= \norm*{ \prn*{\sum_{i=1}^{N} \lambda_i x'_i x_i^{\top}} \prn*{\sum_{i=1}^{N} \lambda_i x_i x_i^{\top}}^{-1} - \As } \label{eq:OLS_like_concentration:e1-1}\\
    &= \norm*{ \prn*{\sum_{i=1}^{N} \lambda_i (x'_i - \As x_i) x_i^{\top}} \prn*{\sum_{i=1}^{N} \lambda_i x_i x_i^{\top}}^{-1} } \label{eq:OLS_like_concentration:e1-2}\\
    &= \norm*{ \prn*{\sum_{i=1}^{N} \lambda_i \tw_i x_i^{\top}} \prn*{\sum_{i=1}^{N} \lambda_i x_i x_i^{\top}}^{-1} } \label{eq:OLS_like_concentration:e1-3}\\
    &\leq \frac{\hilambda}{\lolambda} \norm*{ \prn*{\sum_{i=1}^{N} \tw_i x_i^{\top}} \prn*{\sum_{i=1}^{N} x_i x_i^{\top}}^{-1} } \label{eq:OLS_like_concentration:e1-4}\\
    &\leq \frac{\hilambda}{\lolambda} \norm*{ \prn*{\sum_{i=1}^{N} \tw_i x_i^{\top}} \prn*{\sum_{i=1}^{N} x_i x_i^{\top}}^{-1/2} } \cdot \norm*{ \prn*{\sum_{i=1}^{N} x_i x_i^{\top}}^{-1/2} } \label{eq:OLS_like_concentration:e1-5}\\
    &= \frac{\hilambda}{\lolambda} \cdot \frac{\norm*{ \prn*{\sum_{i=1}^{N} \tw_i x_i^{\top}} \prn*{\sum_{i=1}^{N} x_i x_i^{\top}}^{-1/2} }}{\sigma_{\min}\prn*{ \sum_{i=1}^{N} x_i x_i^{\top} }^{1/2}}, \label{eq:OLS_like_concentration:e1-6}
  \end{align}
  \end{subequations}
  where in \eqref{eq:OLS_like_concentration:e1-3} we introduce the scaled noise $\tw_i := \Ss^{1/2} w_i = x'_i - \As x_i$, which is sub-Gaussian with variance proxy $\prn[\big]{ \norm{\Ss^{1/2}} \sigma }^2$ by \Cref{thm:lemma-subG_linear_transform}; in \eqref{eq:OLS_like_concentration:e1-4} we use the facts that $\sum_{i=1}^{N} \lambda_i x_i x_i^{\top} \succeq \lolambda \sum_{i=1}^{N} x_i x_i^{\top}$ when $\lambda_i \geq \lolambda$ ($\forall i \in [N]$), and that $\norm{AB} \geq \norm{AC}$ for any $A$ when $B \succeq C$ are both symmetric; in \eqref{eq:OLS_like_concentration:e1-5} we apply the sub-multiplicativity of matrix norms.
  
  Now it only suffices to bound the numerator and denominator in \eqref{eq:OLS_like_concentration:e1-6}. For the numerator, using \Cref{thm:lemma-auto_concentration_XXT}, we know that with probability at least $1-\delta$, with sufficiently large $N \geq N_0$, where $N_0 = \tilde{\mathrm{O}}\prn*{ \prn[\big]{ \norm{\Ss^{1/2}} \sigma }^2 d \log \tfrac{1}{\delta} }$,
  we have
  \begin{equation*}
    \sum_{i=1}^{N} x_i x_i^{\top} \succeq \frac{\prn[\big]{ \norm{\Ss^{1/2}} \sigma }^2}{4} NI,
  \end{equation*}
  and thus
  \begin{equation*}
    \sigma_{\min}\prn*{ \sum_{i=1}^{N} x_i x_i^{\top} } \geq \frac{\prn[\big]{\norm{\Ss^{1/2}} \sigma}^2}{4} N.
  \end{equation*}
  Meanwhile, the denominator is bounded by \Cref{thm:lemma-auto_concentration_upper}, so that
  \begin{equation*}
    \norm*{\hA - \As} \lesssim \frac{2 \hilambda}{\norm{\Ss^{1/2}} \sigma \lolambda} \sqrt{\frac{d \log(d/\delta)}{N}}.
  \end{equation*}
  This completes the proof.
\end{proof}

\begin{lemma}\label{thm:concentration_L}
  Under the premises of \Cref{thm:sample_complexity_A}, and Assumptions \ref{assum:sub_gaussian_noise}, \ref{assum:stability} and \ref{assum:bounded_Theta}, for any $\delta \in (0,1)$, with probability at least $1-\delta$,
  \begin{equation*}
    \abs[\big]{\hL(A, \Sg) - \L(A, \Sg)} \lesssim \frac{d \log(dN/\delta) \log(1/\delta)}{N}.
  \end{equation*}
\end{lemma}

\begin{proof}
  We start by decomposing the error into three terms
  \begin{align*}
    \hL(A, \Sg) - \L(A, \Sg) &= \hL(A, \Sg) - \hL(\As, \Sg) \\
    &\quad{}+ \hL(\As, \Sg) - \L(\As, \Sg) \\
    &\quad{}+ \L(\As, \Sg) - \L(A, \Sg).
  \end{align*}
  For the first term, with probability at least $1 - \delta/4$, we have
  \begin{align*}
    &\abs[\big]{\hL(A, \Sg) - \hL(\As, \Sg)} \\
    ={}& \frac{1}{N} \sum_{i=1}^{N} \prn*{ \ell\prn[\big]{ \Sg^{-1/2}(x'_i - A x_i), \Sg } - \ell\prn[\big]{ \Sg^{-1/2}(x'_i - \As x_i), \Sg } } \\
    \leq{}& \frac{1}{2} L \norm[\big]{\Sg^{1/2} (A-\As) x_i}^2
    \leq \frac{1}{2} L \norm{\Sg^{1/2}}^2 \norm{A-\As}^2 \norm{x_i}^2 \\
    \leq{}& \frac{1}{2} L C_x^2 C_{\Sg}^2 \norm{A - \As}^2 \lesssim \frac{d \log(d/\delta)}{N},
  \end{align*}
  where we plug in \Cref{assum:bounded_Theta}, \Cref{assum:apdx-loss_condition} and \Cref{thm:sample_complexity_A}. Similarly, for the third term, with probability at least $1 - \delta/4$,
  \begin{equation*}
    \abs[\big]{\L(A, \Sg) - \L(\As, \Sg)} \leq \frac{1}{2} L C_x^2 C_{\Sg}^2 \norm{A - \As}^2 \lesssim \frac{d \log(d/\delta)}{N}.
  \end{equation*}
  For the second term, note that with the true $\As$, $\Ss^{-1/2} (x'_i - \As x_i) = w_i$ are i.i.d. samples from the base distribution of zero mean. Therefore, using \Cref{thm:lemma-subG_concentration}, we have
  \begin{equation*}
    \norm{w_i} \leq \sqrt{2c \sigma^2 \log \frac{8N}{\delta}},
  \end{equation*}
  with probability at least $1-\delta/(4N)$. Consequently,
  \begin{align*}
    \abs[\big]{\ell(w_i, \Sg) - \E[w]{\ell(w, \Sg)}}
    &\leq \abs[\big]{\ell(w_i, \Sg) - \ell(0, \Sg)} + \abs[\big]{\E[w]{\ell(w, \Sg)}} \\
    &\leq L c \sigma^2 \log \frac{8N}{\delta} + C_{\ell}
  \end{align*}
  by \Cref{assum:apdx-loss_condition}. Finally, using Hoeffding's inequality, we have
  \begin{align*}
    \abs[\big]{\hL(\As, \Sg) - \L(\As, \Sg)}
    &\leq \frac{1}{N} \sum_{i=1}^{N} \abs[\big]{\ell(w_i, \Sg) - \E[w]{\ell(w, \Sg)}} \\
    &\leq \frac{2Lc \sigma^2}{N} \log \frac{8N}{\delta} \log \frac{8}{\delta} + \frac{2C_{\ell}}{N} \log \frac{8}{\delta}
  \end{align*}
  with probability at least $1-\delta/4$. The proof is now completed by applying union bound to the inequalities above.
\end{proof}

\begin{lemma}\label{thm:sample_complexity_Sigma}
  Under the premises of \Cref{thm:concentration_L}, for any $\delta \in (0,1)$, with probability at least $1-\delta$, we have
  \begin{equation*}
    \norm[\big]{\hS - \Ss} \lesssim \sqrt{\frac{d \log(dN/\delta) \log(1/\delta)}{N}}.
  \end{equation*}
\end{lemma}

\begin{proof}
  Note that, by definition of $(\hA, \hS)$ and $(\As, \Ss)$, we have
  \begin{equation*}
    \hL(\hA, \hS) \geq \hL(\As, \Ss),\quad
    \L(\As, \Ss) \geq \L(\hA, \hS).
  \end{equation*}
  Therefore, we shall apply \Cref{thm:concentration_L} twice to obtain
  \begin{align*}
    0 \leq{}& \hL(\hA, \hS) - \hL(\As, \Ss) + \L(\As, \Ss) - \L(\hA, \hS) \\
    \leq{}& \abs[\big]{\hL(\hA, \hS) - \L(\hA, \hS)} + \abs[\big]{\L(\As, \Ss) - \hL(\As, \Ss)} \\
    \lesssim{}& 2 \cdot \frac{d \log(dN/\delta) \log(1/\delta)}{N}
  \end{align*}
  with probability at least $1-\delta$. Meanwhile, by $\mu$-strong convexity assumed in \Cref{assum:apdx-loss_condition}, we have in the opposite direction,
  \begin{align*}
    \hL(\hA, \hS) - \hL(\As, \Ss) &\geq \frac{\mu}{2} \norm{\hS - \Ss}^2, \\
    \L(\As, \Ss) - \L(\hA, \hS) &\geq \frac{\mu}{2} \norm{\hS - \Ss}^2.
  \end{align*}
  Combining these inequalities yields an upper bound on $\norm[\big]{\hS - \Ss}$:
  \begin{equation*}
    \norm[\big]{\hS - \Ss} \lesssim \sqrt{\frac{d \log(dN/\delta) \log(1/\delta)}{N}}.
  \end{equation*}
  This completes the proof.
\end{proof}

To conclude \Cref{thm:main-sample_complexity} as direct corollaries of \Cref{thm:sample_complexity_A} and \Cref{thm:sample_complexity_Sigma}, it only suffices to verify Assumptions \ref{assum:apdx-loss_condition} and \ref{assum:apdx-reweighted_OLS}, which is straight-forward due to the boundedness of $\varTheta$ and the guaranteed finite first-order moment of the base distribution.

\section{Technical Tools}

\boldtitle{Matrix Differential Calculus.} The following formulae for differentiating matrix-variate functions are well-known in literature (see, e.g., \cite{petersen2008matrix}) and can be easily verified.
\begin{subequations}
\begin{align}
  \frac{\partial \log \det(A)}{\partial A} &= A^{-1},\label{eq:matrix_diff-log_det}\\
  \frac{\partial \tr(A)}{\partial A} &= I,\label{eq:matrix_diff-tr_A}\\
  \frac{\partial f(CAx)}{\partial A} &= C^{\top} [\nabla f(CAx)] x^{\top},\label{eq:matrix_diff-CAx}\\
  \frac{\partial f(Ax+b)}{\partial x} &= A^{\top} [\nabla f (Ax+b)],\label{eq:matrix_diff-Ax+b}\\
  \frac{\partial^2 f(Ax+b)}{\partial x^2} &= A^{\top} [\nabla^2 f(Ax+b)] A,\label{eq:matrix_diff-Ax+b_2}\\
  \frac{\partial f(x^{\top} x)}{\partial x} &= 2 f'(x^{\top} x) x,\label{eq:matrix_diff-xTx}\\
  \frac{\partial^2 f(x^{\top} x)}{\partial x^2} &= 2 f'(x^{\top} x) I + 4f''(x^{\top} x) xx^{\top}.\label{eq:matrix_diff-xTx_2}
\end{align}
\end{subequations}


\boldtitle{Sub-Gaussian Random Variable.} Sub-Gaussian random variables, a notion that intuitively captures distributions with tails ``lighter'' than Gaussian, shall be defined as follows.

\begin{definition}[Sub-Gaussian]\label{def:sub_gaussian}
  A random vector $x \in \R^d$ is said to be sub-Gaussian with \emph{variance proxy} $\sigma^2$, if
  \begin{equation*}
    \E{\exp\prn[\big]{v^{\top} (x-\mu) t}} \leq \exp\prn[\big]{\sigma^2 t^2 / 2},\quad \forall t \in \R, \forall v \in \mathbb{S}^{d-1},
  \end{equation*}
  where $\mu := \E{x}$ is the expectation of $x$.
\end{definition}

As its name suggests, sub-Gaussian random variables have tails ``lighter'' than Gaussian, as summarized below.

\begin{lemma}\label{thm:lemma-subG_concentration}
  Given a sub-Gaussian random vector $x \in \R^d$ with variance proxy $\sigma^2$, we have
  \begin{equation*}
    \Prob{\norm{x - \mu} \geq t} \leq 2 \exp\prn*{- \frac{t^2}{2c \sigma^2}}
  \end{equation*}
  for some universal constant $c$, where $\mu := \E{x}$.
\end{lemma}

The following lemma guarantees that the linear transform of a sub-Gaussian random variable is still sub-Gaussian.

\begin{lemma}\label{thm:lemma-subG_linear_transform}
  Given a sub-Gaussian random vector $x \in \R^d$ with variance proxy $\sigma^2$, for any $A \in \R^{d \times d}$, $Ax$ is also a sub-Gaussian random vector, with variance proxy $\prn[\big]{ \norm{A} \sigma }^2$.
\end{lemma}

\begin{proof}
  Let $\mu := \E{x}$. Then for any $t \in \R$ and $v \in \mathbb{S}^{d-1}$ we have
  \begin{align*}
    \E{\exp\prn[\big]{ v^{\top} (Ax - \E{Ax}) t }}
    &= \E{\exp\prn[\big]{ (v^{\top} A(x - \mu) t }} \\
    &= \E{\exp\prn*{ \prn*{\tfrac{A^{\top} v}{\norm{A^{\top} v}} }^{\top} (x - \mu) \cdot (\norm{A^{\top} v} t) }} \\
    &\leq \exp\prn[\big]{ \sigma^2 (\norm{A^{\top} v} t)^2 / 2 } \\
    &\leq \exp\prn[\big]{ (\norm{A} \sigma)^2 t^2 / 2 },
  \end{align*}
  where in the last inequality we use the fact that $\norm{A^{\top} v} \leq \norm{A}$.
\end{proof}

\boldtitle{Concentration Bounds for Time Series.} The following auto-regressive concentration bounds are established in \cite{sarkar2019near}.

\begin{lemma}\label{thm:lemma-auto_concentration_XXT}
  Suppose $z_{i+1} = \As z_i + \eta_i$, where $\rho(\As) < 1$ is stable, and $\set{\eta_i}$ is i.i.d. sub-Gaussian with variance proxy $\sigma^2$. Then, with probability at least $1-\delta$, there exists some
  \begin{equation*}
    N_0 = \tilde{\mathrm{O}}\prn*{\sigma^2 d \log \tfrac{1}{\delta}},
  \end{equation*}
  such that for any $N > N_0$, we have
  \begin{equation*}
    \sum_{i=1}^{N} z_i z_i^{\top} \succeq \frac{\sigma^2 N}{4} I.
  \end{equation*}
\end{lemma}

\begin{proof}
  See Section 9 in \cite{sarkar2019near}.
\end{proof}

\begin{lemma}\label{thm:lemma-auto_concentration_upper}
  Suppose $z_{i+1} = \As z_i + \eta_i$, where $\rho(\As) < 1$ is stable, and $\set{\eta_i}$ is i.i.d. sub-Gaussian with variance proxy $\sigma^2$. Further, suppose the Gramian $\varGamma_t(A) := \sum_{\tau=0}^{t-1} A^{\tau} (A^{\tau})^{\top}$ is bounded by $\varGamma$. Then, with probability at least $1-\delta$, we have
  \begin{equation*}
    \norm*{ \prn*{\sum_{i=1}^{N} \eta_i z_i^{\top}} \prn*{\sum_{i=1}^{N} z_i z_i^{\top}}^{-1/2} }
    \lesssim \sqrt{d \log(d/\delta)}.
  \end{equation*}
\end{lemma}

\begin{proof}
  For simplicity, let $S := \sum_{i=1}^{N} \eta_i z_i^{\top}$ and $Y := \sum_{i=1}^{N} z_i z_i^{\top}$. Define $\lambda := \sigma_{\min}(Y)$, such that $\bar{Y} := Y + \lambda I \preceq 2Y$. Then by Proposition 8.2 in \cite{sarkar2019near} we have
  \begin{equation*}
    \norm[\big]{SY^{-1/2}}
    \leq \sqrt{2} \norm[\big]{S \bar{Y}^{-1/2}}
    \leq 4\sigma \sqrt{d \log\prn*{\frac{5 \prn*{\frac{\det(\bar{Y})}{\det(\lambda I)}}^{1/2d}}{\delta}}}.
  \end{equation*}
  Meanwhile, by Proposition 8.4 in \cite{sarkar2019near} we also have
  \begin{equation*}
    \norm{Y} \leq \tr\prn*{\sum_{i=1}^{N} \varGamma_t(A)} \prn*{1 + \frac{1}{C} \log \frac{1}{\delta}}
    \leq dN\varGamma \prn*{1 + \frac{1}{C_1} \log \frac{1}{\delta^{1/d}}}
  \end{equation*}
  for some absolute constant $C_1$. Consequently, by \Cref{thm:lemma-auto_concentration_XXT},
  \begin{equation*}
    \frac{\det(\bar{Y})}{\det(\lambda I)}
    \leq 2 \prn*{\frac{\norm{Y}}{\lambda}}^d
    \leq 2 \prn*{\frac{4d\varGamma \prn*{1 + \frac{1}{C_1} \log \frac{1}{\delta}}}{\sigma^2}}^d,
  \end{equation*}
  and thus
  \begin{equation*}
    \norm[\big]{SY^{-1/2}}
    \leq 4\sigma \sqrt{\frac{d}{2} \log\prn*{ \frac{Cd \log (1/\delta)}{\delta}} }
    \asymp \sqrt{d \log(d / \delta)}.
  \end{equation*}
  for some absolute constant $C$.
\end{proof}